\documentclass[10pt,journal,compsoc]{IEEEtran}

\ifCLASSOPTIONcompsoc
  \usepackage[nocompress]{cite}
\else
  \usepackage{cite}
\fi

\ifCLASSINFOpdf
\else
\fi

\ifCLASSOPTIONcompsoc

\else

\fi

\ifCLASSINFOpdf

\else

\fi
\usepackage{xspace}
\usepackage{balance}
\usepackage{algorithm}
\usepackage[center]{caption2}
\usepackage{algorithmic}
\usepackage{diagbox}
\usepackage{bbding}
\usepackage{amsthm}
\usepackage{amsmath,amssymb,amsfonts}
\usepackage{makecell}
\usepackage{diagbox}
\usepackage{multirow}
\usepackage{threeparttable}
\usepackage{subfigure}
\usepackage{url}
\usepackage{enumerate}
\usepackage{booktabs}
\usepackage{pifont}
\usepackage{tikz}
\usepackage{etoolbox}
\newcommand{\Name}{\textit{VerifyML}\xspace}

\newenvironment{packeditemize}{
\begin{list}{$\bullet$}{
\setlength{\labelwidth}{8pt}
\setlength{\itemsep}{0pt}
\setlength{\leftmargin}{\labelwidth}
\addtolength{\leftmargin}{\labelsep}
\setlength{\parindent}{0pt}
\setlength{\listparindent}{\parindent}
\setlength{\parsep}{0pt}
\setlength{\topsep}{3pt}}}{\end{list}}
\newtheorem{theorem}{Theorem}[section]

\renewcommand{\raggedright}{\leftskip=0pt \rightskip=0pt plus 0cm}
\raggedright
\begin{document}
\title{VerifyML: Obliviously Checking Model Fairness Resilient to Malicious Model Holder}

\author{{Guowen~Xu, Xingshuo~Han, Gelei~Deng, Tianwei~Zhang, Shengmin~Xu, Jianting~Ning, Anjia~Yang, Hongwei~Li}
\IEEEcompsocitemizethanks{\IEEEcompsocthanksitem Guowen~Xu,Xingshuo~Han,Gelei~Deng and Tianwei~Zhang are with the School of Computer Science and Engineering, Nanyang Technological University. (e-mail: guowen.xu@ntu.edu.sg; xingshuo001@e.ntu.edu.sg; GDENG003@e.ntu.edu.sg; tianwei.zhang@ntu.edu.sg)
\IEEEcompsocthanksitem Shengmin~Xu  and Jianting~Ning are with the College of Computer and Cyber Security, Fujian
Normal University, Fuzhou, China (e-mail: smxu1989@gmail.com; jtning88@gmail.com)
\IEEEcompsocthanksitem  Anjia~Yang  is  with the College of Cyber Security, Jinan University, Guangzhou 510632, China.
(e-mail:anjiayang@gmail.com)
\IEEEcompsocthanksitem Hongwei~Li  is with the school of Computer Science and Engineering,  University of Electronic Science and Technology of China, Chengdu 611731, China.(e-mail: hongweili@uestc.edu.cn)}}

 \IEEEcompsoctitleabstractindextext{
\begin{abstract}
 \raggedright
 In this paper, we present \Name, the first secure inference framework  to check the fairness degree of a given Machine learning (ML) model. \Name  is generic and is immune to any obstruction by the malicious model holder during the verification process. We rely on secure two-party computation (2PC)  technology to implement \Name, and carefully customize a series of optimization methods to boost its performance for both linear and nonlinear layer execution. Specifically, (1) \Name allows the vast majority of the overhead to be performed offline, thus meeting the low latency requirements for online inference. (2) To speed up offline preparation, we first design novel homomorphic parallel computing techniques to accelerate the authenticated Beaver's triple (including matrix-vector and convolution triples) generation procedure. It achieves up to $1.7\times$ computation speedup and gains at least $10.7\times$ less communication overhead compared to state-of-the-art work. (3) We also present a new  cryptographic protocol to evaluate the activation functions of non-linear layers, which is $4\times$--$42\times$ faster  and has $>48\times$ lesser communication than existing 2PC  protocol against malicious parties. In fact,  \Name even beats the state-of-the-art semi-honest ML secure inference system! We provide formal theoretical analysis for \Name security and demonstrate its performance superiority on  mainstream ML models including ResNet-18 and LeNet.
\end{abstract}
}
\maketitle

\IEEEdisplaynotcompsoctitleabstractindextext

\IEEEpeerreviewmaketitle

\section{Introduction}
Machine learning (ML) systems are increasingly being used to inform and influence people's decisions, leading to algorithmic outcomes that have powerful implications for individuals and society. For example, most personal loan default risks are calculated by automated ML tools. This approach greatly speeds up the decision-making process, but as with any decision-making algorithm, there is a tendency to provide accurate results for the majority, leaving certain individuals and minority groups  disadvantaged  \cite{venturebeat,lexology}. This problem is widely defined as the unfairness of the ML model. It often stems from the underlying inherent human bias in the training samples, and a trained ML model amplifies this bias to the point of causing discriminatory decisions about certain groups and individuals.

 Actually, the unfairness of ML model entangles in every corner of society, not only being spied on in financial risk control. A prime example comes from COMPAS \cite{propublica}, an automated software used in US courts to assess the probability of criminals reoffending. A investigation of the software reveals a bias against African-Americans, i.e., COMPAS having a higher false positive rate for African-American offenders than white criminals, owing to incorrectly estimating their risk of reoffending.  Similar model decision biases pervade other real-world applications including childcare systems \cite{chouldechova2018case}, employment matching \cite{osoba2017intelligence}, AI chatbots, and ad serving algorithms \cite{howard2018ugly}. As mentioned earlier, these resulting unfair decisions stem from neglected biases and discrimination hidden in data and algorithms.

 To alleviate the above problems,  a series of recent works \cite{biswas2020machine,lahoti2020fairness,prost2021measuring,oneto2020fairness,mukherjee2020two} have proposed for formalizing measures of fairness for classification models, as well as their variants, in aim to provide instructions for verifying the fairness of a given model. Several evaluation tools have also been released that facilitate automated checks for discriminatory decisions in a given model. For example, Aequitas \cite{saleiro2018aequitas} as a toolkit provides testing of models against several bias and fairness metrics corresponding to different population subgroups. It feeds back test reports to developers, researchers and governments to assist them in making conscious decisions to avoid tending to harm specific population groups. IBM also offers a toolkit AI Fairness 360  \cite{bellamy2018ai}, which aims to bringing fairness research algorithms to the industrial setting, creating a benchmark where all fairness algorithms can be evaluated, and providing an environment for  researchers to share their ideas.

 Existing efforts in theory and tools have led the entire research community to work towards unbiased verification of the ML model fairness. However, existing verification mechanisms either require to white-box access the target model or require clients to send queries in plaintext to the model holder, which is impractical as it incurs a range of privacy concerns. Specifically, model holders are often reluctant to disclose model details because training a commercial model requires a lot of human cost, resources, and experience. Therefore, ML models, as precious intellectual property rights, need to be properly protected to ensure the company's competitiveness in the market. On the other hand, the queries that clients used to test model fairness naturally contain sensitive information, including loan records, disease history, and even criminal information. These highly private data should clearly be guaranteed confidentiality throughout the verification process. Hence, these requirements for privacy raises a challenging but meaningful question:

 \textit{Can we design a verification framework  that only returns the fairness of the model to the client and the parties cannot gain any private information?}

 We materialize the above question to a scenario where a client interacts with the model holder to verify the fairness of the model.  Specifically, before using the target model's inference service, the client sends a set of queries for testing fairness to the model holder, which returns inference results to the client enabling it to locally evaluate how fair the model is.  In such a scenario, the client is generally considered to be semi-honest since it needs to evaluate the model correctly for subsequent service. The model holder may be malicious, it may trick the client into believing that the model is of high fairness by arbitrarily violating the verification process. A natural solution to tackle such concerns is to leverage state-of-the-art generic 2PC  tools \cite{keller2018overdrive,keller2020mp,chen2020maliciously} that provide \textit{malicious security}. It guarantees that if either entity behaves maliciously, they will be caught and the protocol aborted, protecting privacy. However, direct grafting of  standard tools incurs enormous redundant overhead, including heavy reliance on zero-knowledge proofs \cite{feige1988zero}, tedious computational authentication and interaction \cite{hazay2018concretely} (see Section~\ref{Technical Intuition} for more details).

 To reduce the overhead, we propose \Name, a 2PC-based secure verification framework implemented on the  \textit{model holder-malicious} threat model. In  this model, the client is considered semi-honest but the model holder is malicious and can arbitrarily violate the specification of the protocol. We adaptively customize a series of optimization methods for \Name, which show much better performance than the fully malicious baseline. Our key insight is to move the vast majority of operations to the client to bypass cumbersome data integrity verification and reduce the frequency of interactions between entities. Further, we design  highly optimized methods to perform linear and nonlinear layer functions for ML, which brings at least $4-40\times$ speedup compared to state-of-the-art techniques.
 Overall, our contributions are as follows:

\begin{packeditemize}

\item  We leverage the hybrid combination of HE-GC to design \Name. In \Name, the execution of ML's linear layer is implemented by homomorphic encryption (HE) while the non-linear layer is performed by the garbled circuit (GC). VerifyML allows more than 95\% of operations to be completed in the offline phase, thus providing very low latency in the online inference phase. Actually, VerifyML's online phase even beats DELPHI \cite{mishra2020delphi}, the state-of-the-art scheme for secure ML inference against only semi-honest adversaries.

\item We design a series of optimization methods to reduce the overhead of the offline stage. Specifically, we design new homomorphic parallel computation methods, which are used to generate authenticated Beaver's triples, including matrix-vector and convolution triples, in a Single Instruction Multiple Data (SIMD) manner. Compared to existing techniques, we generate triples of matrix-vector multiplication without any homomorphic rotation operation, which is very computationally expensive compared to other homomorphic operations including addition and multiplication. Besides, we reduce the communication complexity of generating convolution triples (aka matrix multiplication triples) from  cubic to quadratic with faster computing performance.

\item  We  design computationally-friendly GC to perform activation functions of nonlinear layers (mainly ReLU). Our key idea is to minimize the number of expensive multiplication operations in the GC. Then, we use the GC as a one-time pad to simplify verifying the integrity of the input from the server. Compared to the state-of-the-art works, our non-linear layer protocol achieves at least an order of magnitude performance improvement.

 \item We provide formal theoretical analysis for \Name security and demonstrate its performance superiority on various datasets  and mainstream ML models including ResNet-18 and LeNet. Compared to state-of-the-art work,  our experiments show that \Name  achieves up to $ 1.7\times$ computation speedup and gains at least  $10.7\times$ less communication overhead for linear layer computation.  For non-linear layers, \Name is  also $4\times$--$42\times$ faster  and has $>48\times$ lesser communication than existing 2PC  protocol against malicious parties. Meanwhile, \Name demonstrates an encouraging online runtime boost by $32.6\times$ and  $32.2\times$ over existing works  on LeNet and ResNet-18, respectively,  and at least  an order of magnitude  communication cost reduction.
\end{packeditemize}

%\textbf{Roadmap}: The remainder of this paper is organized as follows. In  Section \ref{sec:PROBLEM STATEMENT}, we review some basic concepts and introduce the scenarios and threat models involved in this article. In  Section \ref{sec:The VerifyML Framework}, we give the details of our  \Name.  Next, performance evaluation  is presented in Section \ref{sec:performance evaluation}. Finally, Section \ref{sec:conclusion} concludes the paper.

\section{Preliminaries}
\label{sec:PROBLEM STATEMENT}
%In this section, we first describe the threat model in our framework, and then define some notations as well as review several cryptographic primitives used in this paper.

\subsection{Threat Model}
\label{Threat Model}
We consider a secure ML inference scenario, where a model holder $P_0$ and a client $P_1$ interact with each other to evaluate the fairness of the target model. In such a \textit{model holder-malicious} threat model, $P_0$ holds the model $\mathbf{M}$  while the client owns the private test set  used to verify the fairness of the model. The client is generally considered to be semi-honest, that is, it follows the protocol's specifications in the interaction process for evaluating the fairness of the model unbiased. However, it is possible to infer model parameters by passively analyzing data streams captured during interactions. The model holder is malicious. It may arbitrarily violate the specification of the protocol to trick clients into believing that they hold a high-fairness model.  The network architecture is assumed
to be known to both $P_0$ and $P_1$. \Name aims to construct such a secure inference framework that enables $P_1$ to correctly evaluate the fairness of model without knowing any details of the model parameters, meanwhile, $P_0$ knows nothing about the client's input. We provide a formal definition of the threat model in Appendix~\ref{A:threat model}.

\subsection{Notations}
\label{sec:notations}
We use $\lambda$ and $\sigma$ to denote  the computational security parameter and the statistical security parameter, respectively. $[k]$ represents the set $\{1,2, \cdots k\}$ for $k>0$.  In our \Name, all the arithmetic operations are calculated in the field $\mathbb{F}_p$, where $p$ is a  a prime and we define $\kappa=\lceil \log p \rceil$. This means that there is a natural mapping for elements in $\mathbb{F}_p$ to $\{0,1\}^{\kappa}$. For example, $a[i]$ indicates  the $i$-th bit of $a$  on this mapping, i.e, $a=\sum_{i\in[\kappa]}a[i]\cdot 2^{i-1}$. Given two vectors $\mathbf{a}$ and $\mathbf{b}$, and an element $\alpha\in \mathbb{F}_p$, $\mathbf{a}+\mathbf{b}$ indicates the element-wise addition, $\alpha+\mathbf{a}$  and $\alpha\mathbf{a}$ mean that each component of $\mathbf{a}$ performs addition and multiplication with $\alpha$, respectively. $\mathbf{a}\ast \mathbf{b}$ represents the inner production between vectors  $\mathbf{a}$ and $\mathbf{b}$.  Similarly, given any function $f: \mathbb{F}_p \rightarrow \mathbb{F}_p$, $f(\mathbf{a})$ denotes evaluation of $f$ on  each component on $\mathbf{a}$. $a||b$ represents the concatenation of $a$ and $b$. $U_n$ is used to represent the uniform distribution on the set $\{0, 1\}^{n}$ for any $n>0$.

For ease of exposition, we consider an ML model, usually a neural network model $\mathbf{M}$, consisting of alternating linear and nonlinear layers. We assume that the specification of the linear layer is $ \mathbf{L}_1, \cdots \mathbf{L}_{m}$ and the non-linear layer is $f_1, \cdots, f_{m-1}$. Given an initial input (i.e. query) $\mathbf{x}_0$, the model holder will sequentially execute $\mathbf{v}_i=\mathbf{L}_i \mathbf{x}_{i-1}$ and $\mathbf{x}_i=f_i(\mathbf{v}_i)$. Finally, $\mathbf{M}$ outputs the inference result $\mathbf{v}_m=\mathbf{L}_m \mathbf{x}_{m-1}=\mathbf{M}(\mathbf{x}_0)$.

\subsection{ML Fairness Measurement}
Let $\mathcal{X}$ be the set of possible inputs and $\mathcal{Y}$ be the set of all possible labels. In addition, let $\mathcal{O}$ be a finite set related to fairness (e.g., ethnic group). We assume that $\mathcal{X}\times \mathcal{Y}\times\mathcal{O}$ is drawn from a probability space $\Omega$ with an unknown distribution $\mathcal{D}$, and use $\mathbf{M}(\mathbf{x})$ to denote the model inference result given an input $\mathbf{x}$. Based on these, we review the term of the \textit{empirical fairness gap} (EFG) \cite{segal2021fairness}, which is widely used to measure the fairness of ML models against a specific group. To formalize the formulation of EFG, we first describe the definition of \textit{conditional risk} as follows:

\begin{small}
 \begin{equation}
 \label{eq1}
\begin{split}
\digamma_o(\mathbf{M})=\mathop{\mathbb{E}}\limits_{(\mathbf{x}, y, o')\sim \mathcal{D}}[\mathbb{I}\{\mathbf{M}(\mathbf{x})\neq y\}|\mathit{o}'= \mathit{o}]
\end{split}
\end{equation}
\end{small}
Given a set of samples  $(\mathbf{x}, y, o')$ satisfying distribution $\mathcal{D}$, $\digamma_o(\mathbf{M})$ is the expectation of the number of misclassified entries in the test set that belong to group $o$, where $\mathbb{I}\{\Phi\}$ represents the indicator function with a predicate $\Phi$. Given an independent sample set $\Psi=\{ (\mathbf{x}^{(1)}, y^{(1)}, o^{(1)}), \cdots, (\mathbf{x}^{(t)}, y^{(t)}, o^{(t)})\}$$\sim$ $\mathcal{D}^{t}$, the \textit{empirical conditional risk} is defined as follows:

\begin{small}
 \begin{equation}
 \label{eq2}
\begin{split}
\tilde{\digamma}_o(\mathbf{M}, \Psi)=\frac{1}{t_o}\sum_{i=1}^{t}[\mathbb{I}\{\mathbf{M}(\mathbf{x}^{(i)})\neq y^{(i)}\}|\mathit{o}^{(i)}= \mathit{o}]
\end{split}
\end{equation}
\end{small}
where $t_o$ indicates the number of samples in  $\Psi$ from group $o$. Then,  we describe the term \textit{fairness gap (FG)}, which is used to measure the maximum margin of any two groups, specifically,
\begin{small}
 \begin{equation}
 \label{eq3}
\begin{split}
FG=\max_{o_o, o_1\in \mathcal{O}}|\digamma_{o_o}(\mathbf{M})-\digamma_{o_1}(\mathbf{M})|
\end{split}
\end{equation}
\end{small}
Likewise, the \textit{empirical fairness gap (EFG)} is defined as
\begin{small}
 \begin{equation}
 \label{eq4}
\begin{split}
EFG=\max_{o_o, o_1\in \mathcal{O}}|\tilde{\digamma}_{o_o}(\mathbf{M}, \Psi)-\tilde{\digamma}_{o_1}(\mathbf{M}, \Psi)|
\end{split}
\end{equation}
\end{small}

Lastly, we say a ML model $\mathbf{M}$ is $\epsilon$-\textbf{fair} on $(\mathcal{O}, \mathcal{D})$, if its fairness gap is smaller than $\epsilon$ with confidence $1-\delta$. Formally, a $\epsilon$-\textbf{fair} $\mathbf{M}$ is defined as satisfying the following conditions:
\begin{small}
 \begin{equation}
 \label{eq5}
\begin{split}
Pr\left[\max_{o_o, o_1\in \mathcal{O}}|\digamma_{o_o}(\mathbf{M})-\digamma_{o_1}(\mathbf{M})|>\epsilon\right]\leq \delta
\end{split}
\end{equation}
\end{small}
In practice, we usually replace \textit{FG} in Eqn.\ref{eq5} with \textit{EFG} to facilitate the measurement of fairness. Note that once the client gets enough predictions in the target model, it can locally evaluate the fairness of the model according to Eqn.\ref{eq5}.

\subsection{Fully Homomorphic Encryption}
\label{Fully Homomorphic Encryption}
%Fully homomorphic encryption (FHE) enables the evaluation of arbitrary functions (parsed as polynomials) under ciphertext without decryption.
Let the plaintext space be $\mathbb{F}_p$, informally, a Fully homomorphic encryption (FHE) under the public key encryption system usually contains the following algorithms:
\begin{packeditemize}

\item $\mathtt{KeyGen}(1^\lambda)\rightarrow (pk, sk)$. Taking the security parameter $\lambda$ as input, $\mathtt{KeyGen}$ is a random algorithm used to output the public key $pk$ and the corresponding secret key $sk$ required for homomorphic encryption.
\item $\mathtt{Enc}(pk, x)\rightarrow c$.  Given $pk$ and a plaintext $x\in\mathbb{F}_p$, the algorithm $\mathtt{Enc}$ outputs  a ciphertext $c$  encrypting  $x$.
\item $\mathtt{Dec}(sk, c)\rightarrow x$.  Taking $sk$ and a ciphertext $c$ as input, $\mathtt{Dec}$ decrypts $c$ and outputs the corresponding plaintext $x$.
\item $\mathtt{Eval}(pk, c_1, c_2, F)\rightarrow c'$. Given $pk$, two ciphertexts  $c_1$ and $c_2$,  and a function $F$, the algorithm $\mathtt{Eval}$ outputs  a ciphertext $c'$  encrypting  $F(c_1, c_2)$.
\end{packeditemize}
We require FHE to satisfy correctness, semantic security, and functional privacy\footnote{Functional privacy ensures that given a ciphertext $c$, which is an encrypted share of $F(x_1, x_2)$  obtained by homomorphically evaluating $L$, $c$ is indistinguishable from ciphertext $c'$  encrypting a share of $F'(x_1, x_2)$ for any $F'$. }. In \Name, we use the SEAL library \cite{sealcrypto} to implement the fully homomorphic encryption. In addition, we utilize ciphertext packing technology (CPT) \cite{smart2014fully} to encrypt multiple plaintexts  into a single ciphertext, thus enabling homomorphic computation in a SIMD manner.  Specifically, given two plaintext vectors $\mathbf{x}=(x_1, \cdots, x_n)$ and $\mathbf{x'}=(x_1', \cdots, x_n')$, we can pack $\mathbf{x}$ and $\mathbf{x}'$ into ciphertexts $c$ and $c'$ each of them containing $n$ plaintext slots. Homomorphic operations between $c$ and $c'$ including addition and multiplication are equivalent to performing the same element-wise operations on the corresponding plaintext slots. %For example, $\mathtt{Eval}(pk, c+c')$ outputs a ciphertext which encrypts the plaintext vector $\mathbf{x}''=(x_1+x_1', \cdots, x_n+x_n')$.

FHE also provides  algorithm $\mathtt{Rotation}$   to handle operations between data located in different plaintext slots. Informally, given a plaintext vector $\mathbf{x}=(x_0, \cdots, x_n)$ is encrypted into a single ciphertext $c$, $\mathtt{Rotation}(pk, c, j)$ transforms $c$ into another ciphertext $c'$ whose encrypted plaintext vector is $x'=(x_{j+1}, x_{j+2}, \cdots, x_1, \cdots, x_{j})$. In this way, data on different plaintext slots can be moved to the same position to achieve element-wise operations under ciphertext. In FHE, \textit{rotation} operations are computational expensive compared to homomorphic addition and multiplication operations. Therefore, the optimization criterion for homomorphic SIMD operations is to minimize the number of \textit{rotation} operations.

\subsection{Parallel Matrix Homomorphic Multiplication}
\label{Parallel Matrix Homomorphic Multiplication}
We review the parallel homomorphic multiplication method between arbitrary matrices proposed by \textit{Jiang et al.}\cite{jiang2018secure}, which will be used to accelerate the generation of authenticated triples for convolution in \Name. We take the homomorphic multiplication of two $d\times d$ dimensional matrices as an example. Specifically, given a $d\times d$ dimensional matrix $\mathbf{X}={(x_{i,j})}_{0\leq i,j<d}$, we first define four useful permutations, $\sigma$, $\tau$, $\phi$, and $\varphi$, over the field $\mathbb{F}_p^{d\times d}$. Let $\sigma(\mathbf{X})_{i,j}=\mathbf{X}_{i, i+j}$, $\tau(\mathbf{X})_{i,j}=\mathbf{X}_{i+j, j}$, $\phi(\mathbf{X})_{i,j}=\mathbf{X}_{i, j+1}$ and $\varphi(\mathbf{X})_{i,j}=\mathbf{X}_{i+1,j}$. Then for two square matrices $\mathbf{X}$ and  $\mathbf{Y}$ of order $d$, we can calculate the matrix multiplication between the two by the following formula:

\begin{small}
 \begin{equation}
 \label{eqq}
\begin{split}
\mathbf{X}\ast \mathbf{Y}=\sum_{k=0}^{d-1}(\phi^{k}\circ \sigma(\mathbf{X}) )\odot(\varphi^{k}\circ \tau (\mathbf{Y}))
\end{split}
\end{equation}
\end{small}
where $\odot$ denotes the element-wise multiplication. We provide a toy example of the multiplication of two $3\times3$ matrices in Figure~\ref{Fig:matrix multiplication} for ease of understanding.

We can convert a $d\times d$-dimensional matrix to a vector of length $d^2$ by encoding map $\mathbb{F}_p^{d^2}\rightarrow \mathbb{F}_p^{d\times d}$: $\mathbf{x}=(x_0, \cdots, x_{d^2-1})\mapsto \mathbf{X}={(x_{d\cdot i+j})}_{0\leq i,j<d}$. A ciphertext is said to encrypt a matrix $\mathbf{X}$ if it encrypts the corresponding plaintext vector $\mathbf{x}$. Therefore, given two square matrices $\mathbf{X}$ and $\mathbf{X}$, the multiplication of the two under the ciphertext is calculated as follows:
\begin{small}
 \begin{equation}
 \label{eqq}
\begin{split}
\mathbf{c}_\mathbf{X}\circledast \mathbf{c}_\mathbf{Y}=\sum_{k=0}^{d-1}(\phi^{k}(\mathtt{Enc}_{pk}(\sigma (\mathbf{X})) ))\boxtimes(\varphi^{k}(\mathtt{Enc}_{pk}( \tau(\mathbf{Y}))))
\end{split}
\end{equation}
\end{small}
\begin{figure*}[htb]
\centering
\includegraphics[width=0.95\textwidth]{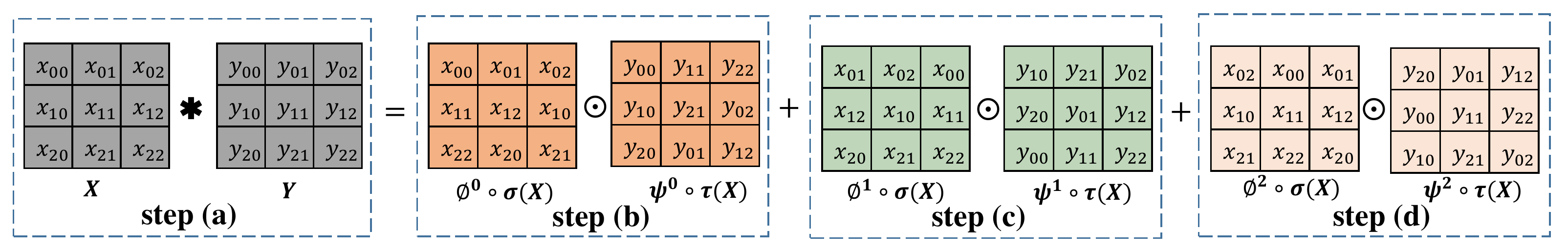}
\caption{Parallel matrix multiplication}
\label{Fig:matrix multiplication}
\vspace{-15pt}
\end{figure*}

 In the following sections, we will  use $\mathbf{c}_\mathbf{X}\circledast \mathbf{c}_\mathbf{Y}$ to represents  multiplication between any matrixes  $\mathbf{X}$  and $\mathbf{Y}$ in ciphertext.  $\boxtimes$ denotes the elewent-wise homomorphic multiplication between two ciphertexts. In Section~\ref{Generating convolution triple}, we describe how to utilize the parallel homomorphic multiplication described above to boost the generation of authenticated convolution triples.
\subsection{Secret Sharing}
\label{Secret Sharing}

\begin{packeditemize}
\item \textbf{Additive Secret Sharing}. Given any $x \in \mathbb{F}_p$, a 2-out-of-2 additive secret sharing of $x$ is a pair $(\left \langle x \right \rangle_0, \left \langle x \right \rangle_1)=(x-r, r)\in \mathbb{F}_{p}^{2}$, where $r$ is a random value uniformly selected from $\mathbb{F}_p$, and $x=\left \langle x \right \rangle_0+\left \langle x \right \rangle_1$.  Additive secret sharing is perfectly hiding, that is, given a share $\left \langle x \right \rangle_0$ or $\left \langle x \right \rangle_1$, $x$ is perfectly hidden.
\item \textbf{Authenticated Shares}. Given a random value $\alpha$  (known as the MAC key)  uniformly chosen  from $\mathbb{F}_p$, for any $x \in \mathbb{F}_p$, the authenticated shares  of $x$ on $\alpha$ denote that each party $P_b$ holds $[\![x]\!]_b=\{ \langle \alpha \rangle_b, \langle x\rangle_b, \langle \alpha x\rangle_b\}_{b\in \{0, 1\}}$\footnote{Sometimes in $[\![x]\!]_b$ we omit $\langle \alpha \rangle_b$ for brevity.}, where we have $(\langle \alpha \rangle_0 +\langle \alpha \rangle_1)\times (\langle x\rangle_0+ \langle x\rangle_1)=(\langle \alpha x\rangle_0+\langle \alpha x\rangle_1)$. While in the general  malicious 2PC setting, $\alpha$ should be generated randomly through interactions between all parties, in our \textit{ model holder-malicious} model, $\alpha$ can be picked up by  $P_1$ and secretly shared with $P_0$.  Authenticated sharing provides $\lfloor \log p \rfloor$ bits of statistical security. Informally, if a malicious $P_0$ tries to forge the shared $x$ to be $x+\beta$, by tampering with its shares  $(\langle x\rangle_0, \langle \alpha x\rangle_0)$ to $(\langle x \rangle_0+\beta, \langle \alpha x\rangle_0+\beta')$, for non-zero $\{ \beta, \beta'\}$, the probability of parties being authenticated to hold the share of $x+\beta$ (\textit{i.e.}, $\alpha x+\beta'=\alpha(x+\beta)$) is at most $2^{-\lfloor \log p \rfloor}$.
\end{packeditemize}

%\subsection{Beaver's multiplication Triples}
%Beaver's multiplication triple technique \cite{chen2020maliciously} has been widely used in secure multiparty computation to compute multiplications privately with the  semi-honest adversary setting \cite{ mohassel2017secureml, wagh2019securenn}. In 2PC, Beaver's multiplication triple is denoted that  each  $P_b$ holds a tuple $\{\langle x \rangle_b, \langle y\rangle_b, \langle z \rangle_b\}_{b\in \{0, 1\}}$, where $x$, $y$, $z\in \mathbb{F}_p$, and satisfy $xy=z$. Suppose two parties hold the secret share of $c$ and $d$, and are required to obtain the share of the product of $c$ and $d$. To perform this process privately, the parties reveal $c-x$ and $d-y$, and then for each party $P_b$, it locally computes the share of $e=c\cdot d$ as follows:
%
%\begin{small}
% \begin{equation}
% \label{eq6}
%\begin{split}
%\langle e \rangle_b=(c-x)\cdot (d-y)+\langle x \rangle_b\cdot (d-y)+(c-x)\cdot \langle y \rangle_b +\langle z \rangle_b.
%\end{split}
%\end{equation}
%\end{small}
%Beaver's multiplication triple is independent of the user's input in the actual execution of the secure computing protocol, and it can be generated offline (see Section~\ref{sec:The VerifyML Framework}) through the interaction of $P_0$ and $P_1$ to speed up the performance of online secure multiplication computations.
\subsection{Authenticated Beaver's Triples}
\label{Authenticated Beaver¡¯s Triples}
%In our \textit{ model holder-malicious} threat model,  since each party holds  authenticated shares of all private data,
In \Name, we require the  technique of authenticated Beaver's triples to detect possible breaches of the protocol from the malicious model holder. In more detail, authenticated Beaver's multiplication triple  is denoted that  each  $P_b$ holds a tuple $\{[\![x]\!]_b, [\![y]\!]_b, [\![z]\!]_b\}_{b\in \{0, 1\}}$, where $x$, $y$, $z\in \mathbb{F}_p$, and satisfy $xy=z$. Giving  $P_0$ and $P_1$ holding  authenticated shares  of $c$ and $d$,  i.e., ($[\![c]\!]_0$, $[\![d]\!]_0$), ($[\![c]\!]_1$, $[\![d]\!]_1$), respectively,  to compute the authenticated share of the product of $c$ and $d$,  the parties first reveal $c-x$ and $d-y$, and then each party $P_b$  locally computes the authenticated share of $[\![e=c\cdot d]\!]_b$ as follows:

\begin{footnotesize}
 \begin{equation}
 \label{eq6}
\begin{split}
\langle e \rangle_b&=(c-x)\cdot (d-y)+\langle x \rangle_b\cdot (d-y)+(c-x)\cdot \langle y \rangle_b +\langle z \rangle_b\\
\langle \alpha e \rangle_b&=\langle \alpha \rangle_b (c-x)\cdot (d-y)+\langle  \alpha x \rangle_b\cdot (d-y)+(c-x)\cdot \langle \alpha y \rangle_b +\langle  \alpha z \rangle_b
\end{split}
\end{equation}
\end{footnotesize}

Authenticated  Beaver's multiplication triple is independent of the user's input in the actual execution of the secure computing protocol, thus can be generated offline (see Section~\ref{sec:The VerifyML Framework}) to speed up the performance of online secure multiplication computations. Inspired by existing work to construct custom triples for specific mathematical operations  \cite{ mohassel2017secureml} for improving performance, we generalize traditional Beaver's triples to matrix-vector multiplication and convolution domains.  We provide the definitions of matrix-vector and convolution triples below and leave the description of generating  them to Section~\ref{sec:The VerifyML Framework}.
%However, studies \cite{ mohassel2017secureml} have shown that it is not resource-friendly to generate independent authenticated  Beaver's multiplication triple for each single multiplication operation, since modern ML models are dominated by intensive multiplication operations.  To alleviate it, we propose to generalize the above notion to  matrix-vector and convolution triples and generate such triples offline, thereby substantially reducing resources involved in both of offline and online process.
\begin{packeditemize}

\item \textbf{ Authenticated Matrix-Vector  triples}: is denoted that  each  $P_b$ holds a tuple $\{[\![\mathbf{X}]\!]_b, [\![\mathbf{y}]\!]_b, [\![\mathbf{z}]\!]_b\}_{b\in \{0, 1\}}$, where $\mathbf{X}$  is a matrix uniformly chosen from $\mathbb{F}_p^{d_1\times d_2}$, $\mathbf{y}$ represents a vector selected from  $\mathbb{F}_p^{d_2}$, and  $\mathbf{z}\in \mathbb{F}_p^{d_1}$ satisfying $\mathbf{X}\ast \mathbf{y}=\mathbf{z}$, where $d_1$ and  $d_2$ are determined depending on the ML model architecture.
\item \textbf{Authenticated Convolution triples} (aka matrix multiplication triples\footnote{We can reduce the convolution operation to matrix multiplication by transforming the inputs of convolution appropriately. We provide a detailed description in Section~\ref{sec:The VerifyML Framework}.}):  is denoted that  each  $P_b$ holds a tuple $\{[\![\mathbf{X}]\!]_b, [\![\mathbf{Y}]\!]_b, [\![\mathbf{Z}]\!]_b\}_{b\in \{0, 1\}}$, where $\mathbf{X}$   and $\mathbf{Y}$ are tensors  uniformly chosen from $\mathbb{F}_p^{u_w\times u_h\times c_i}$ and  $\mathbb{F}_p^{(2l+1)\times (2l+1)\times c_i\times c_o}$, respectively. $\mathbf{Z}\in \mathbb{F}_p^{u_w'\times u_h'\times c_o}$ satisfying convolution $\mathtt{Conv}(\mathbf{X}, \mathbf{Y})=\mathbf{Z}$, where  $u_w$,  $u_w$, $u_w'$, $u_h'$, $l$, $c_i$ and $c_o$ are determined depending on the model architecture.
\end{packeditemize}

\subsection{Oblivious Transfer}
\label{Oblivious Transfer}
We take OT$_{n}$ to denote the 1-out-of-2 Oblivious Transfer (OT) \cite{grilo2021oblivious,dottling2020two}. In OT$_{n}$,  the inputs of the sender (assuming $P_0$ for convenience) are two strings $s_0, s_1 \in\{0, 1 \}^{n}$, and the input of the receiver ($P_1$) is a  bit $b\in \{0, 1\}$ for selection. At the end of the OT-execution, $P_1$ learns  $s_b$ while $P_0$ learns  nothing.    In this paper, we require that the instance of OT$_{n}$ is secure against a semi-honest sender and a malicious receiver.  We use OT$_{n}^{\kappa}$ to represent $\kappa$ instances of OT$_{n}$. We exploit \cite{keller2015actively} to implement OT$_{n}^{\kappa}$  with the communication complexity of $\kappa{\lambda+2n}$ bits.

\subsection{Garbled Circuits}
\label{garbled circuits}
The garbling scheme \cite{rosulek2021three,ciampi2021threshold} for boolean circuits parsing arbitrary functions  consists of a pair of algorithms ($\mathtt{Garble}$, $\mathtt{GCEval}$) defined as follows:
\begin{packeditemize}
\item $\mathtt{Garble}(1^\lambda, C)\rightarrow (\mathtt{GC}, \{ \{\mathtt{lab}_{i,j}^{in}\}_{i\in[n]},\{\mathtt{lab}_{j}^{out}\}\}_{j\in\{0,1\}})$. Giving the security parameter $\lambda$ and an arbitrary Boolean circuit $C: \{0,1 \}^{n}\rightarrow \{0, 1\}$, the algorithm $\mathtt{Garble}$ outputs a garbled circuit $\mathtt{GC}$, a set of input labels  $\{\mathtt{lab}_{i,j}^{in}\}_{i\in[n], j\in\{0,1\}}$  of this $\mathtt{GC}$, and a set of output labels  $\{\mathtt{lab}_{j}^{out}\}_{j\in\{0,1\}}$, where the size of each label is $\lambda$ bits. For any $x\in \{0,1\}^{n}$, we refer to $\{\mathtt{lab}_{i,x[i]}^{in}\}_{i\in[n]}$ as the \textit{garbled input} of $x$, and $\mathtt{lab}_{ C(x)}^{out}$ as the \textit{garbled output} of $C(x)$.
\item $\mathtt{GCEval}(\mathtt{GC}, \{\mathtt{lab}_{i}\}_{i\in[n]})\rightarrow \mathtt{lab'}$. Giving the garbled circuit $\mathtt{GC}$ and a set of input labels $\{\mathtt{lab}_{i}\}_{i\in[n]}$, the algorithm $\mathtt{GCEval}$  outputs a label $\mathtt{lab'}$.
\end{packeditemize}
Let $\mathtt{Garble}(1^\lambda, C)\rightarrow (\mathtt{GC}, \{ \{\mathtt{lab}_{i,j}^{in}\}_{i\in[n]},\{\mathtt{lab}_{j}^{out}\}\}_{j\in\{0,1\}})$,
the above garbled scheme ($\mathtt{Garble}$, $\mathtt{GCEval}$) is required to satisfy the following properties:
\begin{packeditemize}
\item \textbf{Correctness}.  $\mathtt{GCEval}$ is faithfully performed on the $\mathtt{GC}$ and correctly outputs garbled results when given the garbled input  of $x$. Formally, for any Boolean circuit $C$ and input $x\in \{0,1\}^{n}$, $\mathtt{GCEval}$ holds that
 $$\mathtt{GCEval}(\mathtt{GC}, \{\mathtt{lab}_{i, x[i]}^{in}\}_{i\in[n]})\rightarrow \mathtt{lab}_{ C(x)}^{out}$$
\item \textbf{Security}.  Given $C$, the garbled circuit $\mathtt{GC}$ of $C$ and garbled inputs of any $x\in \{0,1\}^{n}$ can be simulated by a  polynomial probability-time simulator $\mathtt{Sim}$. Formally, for any circuit $C$ and input $x\in \{0,1\}^{n}$, we have $(\mathtt{GC},\{\mathtt{lab}_{i, x[i]}^{in}\}_{i\in[n]})\approx \mathtt{Sim}(1^\lambda, C)$, where $\approx$ indicates computational indistinguishability.
\item \textbf{Authenticity}. This implies that given the garbled input of $x$ and $\mathtt{GC}$, it is infeasible to guess the output label of $1-C(x)$. Formally, for any circuit $C$ and $x\in \{0,1\}^{n}$, we have $\left( \mathtt{lab}_{ 1-C(x)}^{out}|\mathtt{GC},\{\mathtt{lab}_{i, x[i]}^{in}\}_{i\in[n]}\right)\approx U_{\lambda}$.
\end{packeditemize}
Without loss of generality, the garbled scheme described above can be naturally extended to securely implement Boolean circuits with multi-bit outputs. In \Name,   we  utilize state-of-the-art optimization strategies, including point-and-permute \cite{frederiksen2013minilego},
free-XOR \cite{kolesnikov2014flexor} and half-gates \cite{zahur2015two} to construct the garbling scheme.

%As background, we give a high-level description of how to implement generalized 2PC with garbled circuits and OTs as the underlying technology. Assuming  that  a malicious party $P_0$ (with input $x$) interacts with a semi-honest party $P_1$ (with input $y$) to securely evaluate an arbitrary Boolean circuit $C$ to learn $C(x, y)$. $P_1$ first garbles circuit $C$ to learn  $\mathtt{GC}$  and a set of garbled input and output labels. Then, both parties invoke the  OT functionality, where $P_1$ acts as the sender and sends the $\mathtt{GC}$  as well as set of input labels corresponding  to the input wires of $P_0$, while $P_0$ acts as the receiver who inputs $y$ to obtains the garbled input of $y$. In addition, $P_1$ sends the garbled circuit $\mathtt{GC}$, garbled input of $x$ as well as ciphertexts corresponding to every output wire $w$ (representing 0 or 1) of $C$, with $\mathtt{lab}_{w, 0}$ and $\mathtt{lab}_{w, 1}$ as encryption keys.  As a result, taking as inputs the garbled inputs of $x$ and $y$, $P_0$ evaluates $\mathtt{GC}$ and learns the garbled output for $C(x,y)$. Based on the garbled output and ciphertexts sent before for every  output wire, $P_0$ can get the plaintext of garbled $C(x,y)$.  $P_0$ sends $C(x, y)$ along with the hash of the garbled output to $P_1$ for verification.  $P_1$ accepts it if the hash value  corresponds to $C(x, y)$.

\section{Technical Intuition}
\label{Technical Intuition}
 \Name is essentially a 2PC protocol over  the \textit{ model holder-malicious} threat model, where the client unbiasedly learns the inference results on a given test set, thereby faithfully evaluating the fairness of the target model locally. For boosting the performance of the 2PC protocol execution, we customize a series of optimization methods by fully exploring the advantages of cryptographic primitives and their natural ties in inference process. Below we present a high-level technically intuitive overview of VerifyML's design.

\subsection{Offline-Online Paradigm}
\label{Offline-Online Paradigm}
Consistent with state-of-the-art work on the setting of semi-honest models \cite{mishra2020delphi}, \Name is deconstructed into  an offline stage and an online stage, where the preprocessing process of the offline stage is independent of the input of model holders and clients. In this way, the majority ($>95\%$) of the computation can be performed offline to minimize the overhead of the online process. Figure~\ref{Overview of the VerifyML} provides an overview of \Name, where we describe the computational parts required for the offline and online phase, respectively.

 \renewcommand\tablename{Figure}
\renewcommand \thetable{\arabic{table}}
\setcounter{table}{1}
\begin{table}[!htb]
\centering
\small
\begin{tabular}{|p{8.0cm}|}
\Xhline{1pt}
\textbf{Offline Phase}.  This phase the client and model holder pre-compute data in  preparation for subsequent online execution, which is independent of input from all parties. That is, \Name can run this phase without knowing the client's input $\mathbf{x}_0$ and the model holder's input $\mathbf{M}$.
\begin{packeditemize}
\item \textit{Preprocessing for the linear layer}. The Client interacts with the model holder to generate authenticated triples for matrix-vector multiplication  and convolution.
\item \textit{Preprocessing for the nonlinear layer}. The client constructs a garbled circuit $\mathtt{GC}$ for circuit C parsing  ReLU. The client sends $\mathtt{GC}$ and a set of ciphertexts to the model holder for generating the authenticated shares of ReLU's results.
\end{packeditemize}\\
\Xhline{1pt}
\textbf{Online Phase}. This  phase is divided into following parts.
\begin{packeditemize}
\item \textit{Preamble}. The client secretly shares its input $\mathbf{x}_0$ with the model holder, and similarly, the model holder  shares the model parameter $\mathbf{M}$ with the client. Thus both the model holder and the client hold an authenticated share of $\mathbf{x}_0$ and $\mathbf{M}$. Note that the sharing of $\mathbf{M}$ can be done offline, if the model to be verified is knowed in advance.
\item \textit{Layer evaluation}. Let $\mathbf{x}_i$  be the result of evaluating the first $i$ layers of model $\mathbf{M}$  on $\mathbf{x}_0$. At the beginning of the $i+1$-th layer, both the client and the model holder hold an authenticated share about $\mathbf{x}_i$ and the  $i+1$-th layer parameter $\mathbf{L}_{i+1}$, i.e., parties ${P_b}_{\in \{0, 1\}}$ hold $([\![\mathbf{x}_i]\!]_b, [\![\mathbf{L}_{i+1}]\!]_b)$.
 \begin{packeditemize}
\item [1.] \textit{Linear layer }.  The client interacts with the model holder to perform the authenticated shares of $\mathbf{v}_{i+1}=\mathbf{L}_{i+1} \mathbf{x}_{i+1}$, where both parties securely compute matrix-vector multiplication and convolution operations with the aid of triples generated in the precomputing process.

\item [2.] \textit{Nonlinear layer}.  After the linear layer, the two parties hold the authenticated shares of $\mathbf{v}_{i+1}$. The client and the model holder invoke the OT to send the garbled input of $\mathtt{GC}$  to the model holder. The model holder evaluates the $\mathtt{GC}$, and eventually the two parties get authenticated shares of the ReLU result.
 \end{packeditemize}
\item \textit{Consistency check}. The client interacts with the model holder to check any  malicious behavior of the model holder during the entire inference process. The client uses the properties  of the authenticated sharing to construct the consistency check protocol. If consistency passes, the client locally computes the fairness of the target model, otherwise the client outputs abort.
\end{packeditemize}\\
\hline
\end{tabular}
\caption{Overview of the \Name}
\label{Overview of the VerifyML}
%\vspace{-20pt}
\end{table}
\subsection{Linear Layer Optimization}
\label{Linear layer optimization}
 As described in Figure~\ref{Overview of the VerifyML}, we move almost all linear operations into the offline phase, where we construct customized triples for matrix-vector multiplication and convolution to accelerate linear execution. Specifically, 1) we design an efficient construction of matrix-multiplication triples instead of generating  Beaver's multiplication triples for individual multiplications  (see Section~\ref{Generating matrix-vector multiplication triple}).  Our core insight is a new packed homomorphic multiplication method for matrices and vectors. We explore the inherent connection between secret sharing and homomorphic encryption to remove all the rotation operation in parallel homomorphic computation. 2) We extend the idea of generating  matrix multiplicative triples over semi-honest models \cite{ mohassel2017secureml} into convolution domain over the \textit{ model holder-malicious} threat model (see Section~\ref{sec:The VerifyML Framework}). The core of our construction is derived from E2DM \cite{jiang2018secure}, which proposes a state-of-the-art method for parallel homomorphic multiplication between arbitrary matrices. We further optimize  E2DM to achieve at least $2\times$ computational speedup compared to naive use.

 Our optimization technique for linear layer computation exhibits superior advantages compared to state-of-the-art existing methods \cite{keller2018overdrive,keller2020mp}\footnote{Note that several efficient parallel homomorphic computation methods \cite{juvekar2018gazelle,zhang2021gala} with packed ciphertext have been proposed and run on semi-honest  or client-malicious models \cite{mishra2020delphi,lehmkuhl2021muse,chandran2021simc} for secure inference. It may be possible to transfer these techniques to our method to speed up triple's generation, but this is certainly non-trivial and we leave it for future work.}. In more detail, we reduce the communication overhead from cubic to quadratic (both for offline and online phases) compared to Overdrive \cite{keller2018overdrive}, which is the mainstream tool for generating authenticated multiplicative triples on malicious adversary models(see Section~\ref{sec:The VerifyML Framework} for detailed analysis). %For example,  for multiplying two square matrices of size 128,  experimental results show that we reduced the communication cost from $1.54$ GB to $12.46$ MB, an improvement of over two orders of magnitude. In addition, our method also exhibits a substantially reduction in computational overhead.

\subsection{Non-linear Layer Optimization}
\label{non-linear layer optimization}
We use the garbled circuit to achieve secure computation of nonlinear functions (mainly ReLU) in ML models. Specifically,  assumed that $P_0$ and $P_1$ learn the authenticated sharing about $\mathbf{v}_i=\mathbf{L}_i \mathbf{x}_{i-1}$  after executing the $i$-th linear layer, that is, each party $P_b$ holds $[\![\mathbf{v}_i]\!]_b=\{ \langle \alpha \rangle_b, \langle \mathbf{v}_i\rangle_b, \langle \alpha \mathbf{v}_i\rangle_b\}_{b\in \{0, 1\}}$. Then, $\{\langle \mathbf{v}_i\rangle_b\}_{b\in \{0, 1\}}$ will be used as the input of  ReLU (dented as $f_i$  for brevity) in the $i$-th nonlinear layer for both parties learning the authentication sharing about $\mathbf{x}_i=f_i(\mathbf{v}_i)$, i.e., $[\![\mathbf{x}_i]\!]_b$. However, constructing such a satisfactory garbling scheme has the following intractable problems.
\begin{packeditemize}
\item \textit{How to validate input from the malicious model holder}. Since the model holder is malicious, it must be ensured that the input from the model holder in the $\mathtt{GC}$ (i.e. $\langle \mathbf{v}_i\rangle_0$) is consistent with the share obtained by the previous linear layer. In the traditional malicious adversary model \cite{keller2018overdrive,keller2020mp,chen2020maliciously}, a standard approach is to verify the correctness of the authenticated sharing of all inputs from malicious entities in the $\mathtt{GC}$. However, this is very expensive and takes tens of seconds or even minutes to process a ReLU function. It obviously does not meet the practicality of ML model inference because a modern ML model usually contains thousands of ReLU functions.
\item \textit{How to minimize the number of multiplication encapsulated into $\mathtt{GC}$}.  For the $i$-th nonlinear layer, we need to compute the authenticated shares of the ReLU output, i.e. $[\![\mathbf{x}_i]\!]_b=\{ \langle \alpha \rangle_b, \langle \mathbf{x}_i\rangle_b, \langle \alpha \mathbf{x}_i\rangle_b\}_{b\in \{0, 1\}}$. This requires at least two multiplications on the field, if all computations are encapsulated into the $\mathtt{GC}$. Note that  performing arithmetic multiplication operations in the $\mathtt{GC}$ is expensive and requires at least $O(\kappa^2\lambda)$ communication overhead.
\end{packeditemize}

We design novel protocols to remedy the above problems through the following  insights: (1) garbled
circuits already achieve malicious security against garbled circuit evaluators (i.e., the model holder in our setting) \cite{lehmkuhl2021muse}. This means that we only need to construct a lightweight method  to check the consistency between the input of the malicious adversary in the nonlinear layer and the results obtained by the previous linear layer. Then, this method can be integrated with $\mathtt{GC}$ to achieve end-to-end nonlinear  secure computing (see Section~\ref{sec:The VerifyML Framework}). (2)  It is enough to calculate the output label for each bit of  $f_i(\mathbf{v}_i)$'s  share (i.e., $f_i(\mathbf{v}_i)[j]$, for $1\leq j\leq \kappa$) in the GC, rather than obtaining the exact arithmetic share of $f_i(\mathbf{v}_i)$  \cite{chandran2021simc}. Moreover, we can parse ReLU function as $ReLU(\mathbf{v}_i)=\mathbf{v}_i\cdot sign(\mathbf{v}_i)$, where the sign function $sign(\mathbf{v}_i)$ equals 1 if $t\geq 0$ and 0 otherwise. Hence,  we only encapsulate  the non-linear part of $ReLU(\mathbf{v}_i)$ (\textit{i.e.}, $sign(\mathbf{v}_i)$)  into the $\mathtt{GC}$, thereby  substantially minimizing  the number of multiplication operations.

 Compared with works \cite{keller2018overdrive,keller2020mp,chen2020maliciously} with  malicious adversary, \Name reduces the communication overhead of each ReLU function from $2c\lambda+190\kappa\lambda+232\kappa^2$ to $2d\lambda+4\kappa\lambda+6\kappa^2$, where $d\ll c$. Our experiments show that \Name  achieves $4\times$-$42\times$ computation speedup and gains $48\times$ less communication overhead for nonlinear layer computation.

 \textit{Remark 3.1 }. Beyond the above optimization strategies, we also do a series of strategies to reduce the overhead in the implementation process, including  removing the reliance on distributed decryption primitives in previous works \cite{keller2018overdrive,keller2020mp,chen2020maliciously} and minimizing the number of calls to zero-knowledge proofs of ciphertexts. In  the following section, we provide a comprehensive technical description of the proposed method.

\section{The VerifyML Framework}
\label{sec:The VerifyML Framework}
\subsection{Offline Phase}
\label{Offline phase}
In this section, we describe the technical details of \Name. As described above,  \Name is divided intooffline and online phases.  We first describe the operations that need to be precomputed  in the offline phase, including generating matrix-vector multiplications and triples for convolution, and garbled circuits for constructing the objective function. Then, we introduce the technical details of the online phase.
\subsubsection{\quad\; Generating matrix-vector multiplication triple}
\label{Generating matrix-vector multiplication triple}
Figure~\ref{Algorithm of generating matrix-vector multiplication triple} depicts the interaction between the model holder $P_0$ and the client $P_1$ to generate triples of matrix-vector multiplications. Succinctly, $P_0$ first uniformly selects $\langle \mathbf{X}\rangle_0$ and $ \langle \mathbf{y}\rangle_0 $ and sends their encryption to $P_1$, along with  zero-knowledge proofs about these ciphertexts, where $ \langle \mathbf{y}\rangle_0 $ need to be transformed into matrix $ \langle \mathbf{Y}\rangle_0 $  before encryption (step 2 in Figure~\ref{Algorithm of generating matrix-vector multiplication triple}). $P_1$ recovers $\mathbf{X}$ and $\mathbf{Y}$ in ciphertext and then computes $(\langle \alpha \mathbf{X}\rangle_0, \langle \alpha \mathbf{Y}\rangle_0, \langle \alpha \mathbf{Z}\rangle_0, \langle \mathbf{Z} \rangle_0)$ (step 3 in Figure~\ref{Algorithm of generating matrix-vector multiplication triple}). Then it returns the corresponding ciphertexts to $P_0$. $P_0$ decrypts them and computes $\langle \alpha \mathbf{y}\rangle_1$, $\langle \alpha \mathbf{z}\rangle_1$ and $\langle \mathbf{z} \rangle_1$ (step 4 in Figure~\ref{Algorithm of generating matrix-vector multiplication triple}).
\renewcommand\tablename{Figure}
\renewcommand \thetable{\arabic{table}}
\setcounter{table}{2}
\begin{table}[htb]
\centering
\small
\begin{tabular}{|p{8.0cm}|}
\hline \\
 \textbf{Input:} $\{P_b\}_{b\in \{0, 1\}}$ holds $\langle \mathbf{X}\rangle_b$  uniformly chosen from $\mathbb{F}_p^{d_1\times d_2}$, and $ \langle \mathbf{y}\rangle_b $ uniformly chosen from $\mathbb{F}_p^{d_2}$.  In addition, $P_1$ hold a MAC key $\alpha$ uniformly chosen from $\mathbb{F}_p$.\\
 \textbf{Output:} $P_b$ obtains $\{[\![\mathbf{X}]\!]_b, [\![\mathbf{y}]\!]_b, [\![\mathbf{z}]\!]_b\}_{b\in \{0, 1\}}$ where $\mathbf{X}\ast \mathbf{y}=\mathbf{z}$.\\
 \textbf{Procedure}:
 \begin{packeditemize}
    \item[1.]   $P_0$ and  $P_1$ participate in a secure two-party computation such that $P_0$ obtains an FHE public secret key pair ($pk$, $sk$) while $P_1$ obtains the public key $pk$. This process is performed only once.
    \item[2.] $P_0$ first converts $\langle \mathbf{y}\rangle_0 $ into a $d_1\times d_2$-dimensional matrix  $\langle \mathbf{Y}\rangle_0 $ where each row constitutes a copy of $\langle \mathbf{y}\rangle_0 $. Then,
    $P_0$ send the encryptions $c_1\leftarrow \mathtt{Enc}(pk, \langle \mathbf{X}\rangle_0)$  and $c_2\leftarrow \mathtt{Enc}(pk, \langle \mathbf{Y}\rangle_0)$ to $P_1$ along with  zero-knowledge (ZK) proofs of plaintext knowledge of the two ciphertexts \footnotemark[4].
    \end{packeditemize}
 \begin{packeditemize}
    \item[3.] $P_1$  also converts $\langle \mathbf{y}\rangle_1 $ into a $d_1\times d_2$-dimensional matrix  $\langle \mathbf{Y}\rangle_1$ where each row constitutes a copy of $\langle \mathbf{y}\rangle_1$. Then it
     samples $(\langle \alpha \mathbf{X}\rangle_1, \langle \alpha \mathbf{Y}\rangle_1, \langle \alpha \mathbf{Z}\rangle_1, \langle \mathbf{Z} \rangle_1) $ from $\mathbb{F}_{p}^{4\times ({d_1\times d_2})}$.  $P_1$ sends $c_3=\mathtt{Enc}_{pk}( \alpha(\langle \mathbf{X}\rangle_1+\langle \mathbf{X}\rangle_0)-\langle \alpha \mathbf{X}\rangle_1)$, $c_4=\mathtt{Enc}_{pk}( \alpha(\langle \mathbf{Y}\rangle_1+\langle \mathbf{Y}\rangle_0)-\langle \alpha \mathbf{Y}\rangle_1)$, $c_5=\mathtt{Enc}_{pk}(\alpha(\mathbf{X}\odot \mathbf{Y})-\langle \alpha \mathbf{Z}\rangle_1)$, and $c_6=\mathtt{Enc}_{pk}((\mathbf{X}\odot \mathbf{Y})-\langle \mathbf{Z}\rangle_1)$ to $P_0$.
    \item[4.] $P_0$ decrypts $c_3$, $c_4$,  $c_5$ and   $c_6$ to obtain $(\langle \alpha \mathbf{X}\rangle_0, \langle \alpha \mathbf{Y}\rangle_0, \langle \alpha \mathbf{Z}\rangle_0, \langle \mathbf{Z} \rangle_0)$, respectively. Then,  it sums the elements of each row of the matrices $\langle \alpha \mathbf{Y}\rangle_0$\footnotemark[5], $\langle \alpha \mathbf{Z}\rangle_0$ and $\langle \mathbf{Z} \rangle_0$ to form the vectors $\langle \alpha \mathbf{y}\rangle_0$, $\langle \alpha \mathbf{z}\rangle_0$ and $\langle \mathbf{z} \rangle_0$. $P_1$ does the same for $(\langle \alpha \mathbf{Y}\rangle_1, \langle \alpha \mathbf{Z}\rangle_1, \langle \mathbf{Z} \rangle_1)$ to obtain $\langle \alpha \mathbf{y}\rangle_1$, $\langle \alpha \mathbf{z}\rangle_1$ and $\langle \mathbf{z} \rangle_1$.
    \item[5.] $P_b$ outputs $\{[\![\mathbf{X}]\!]_b, [\![\mathbf{y}]\!]_b, [\![\mathbf{z}]\!]_b\}_{b\in \{0, 1\}}$, where $\mathbf{X}\ast \mathbf{y}=\mathbf{z}$.
    \end{packeditemize}\\
\hline
\end{tabular}
\caption{Algorithm $\pi_{Mtriple}$ for generating authenticated matrix-vector multiplication triple}
\label{Algorithm of generating matrix-vector multiplication triple}
%\vspace{-15pt}
\end{table}
\footnotetext[4]{A ZK proof of knowledge for  ciphertexts is used to state that $c_{1}$ and $c_{2}$ are  valid ciphertexts generated from the given FHE cryptosystem. Readers can refer to \cite{keller2018overdrive,chen2020maliciously} for more details.}
\footnotetext[5]{Note that for $\langle \alpha \mathbf{Y}\rangle_0$, we only take the all elements in the first row as $\langle \alpha \mathbf{y}\rangle_0$ by default.The operation for $\langle \alpha \mathbf{Y}\rangle_1$ is the same as above.}
\setcounter{figure}{3}
\setcounter{table}{3}
\begin{figure}[htb]
\centering
\includegraphics[width=0.5\textwidth]{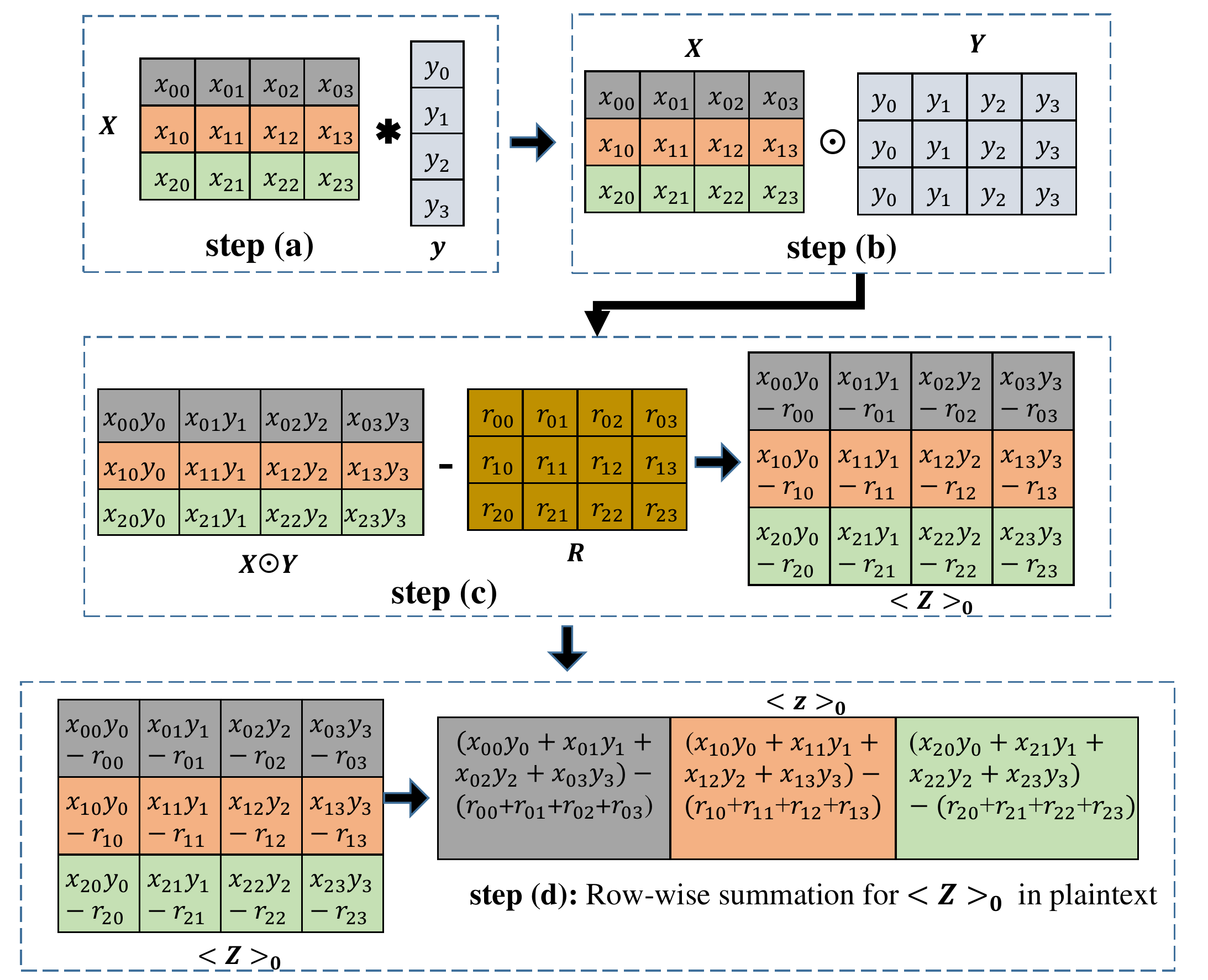}
\caption{Matrix-vector multiplication}
\label{Fig:Matrix-vector multiplication}
%\vspace{-15pt}
\end{figure}

Figure~\ref{Fig:Matrix-vector multiplication} provides an example of the multiplication of a $3\times4$-dimensional matrix $\mathbf{X}$ and a $4$-dimensional vector $\mathbf{y}$ to facilitate understanding. To compute the additive sharing of $\mathbf{z}= \mathbf{X} \ast \mathbf{y}$ (step(a) in Figure~\ref{Fig:Matrix-vector multiplication}), $\mathbf{y}$ is first transformed into a matrix $\mathbf{X}$ by copying, where each row of $\mathbf{Y}$ contains a copy of $\mathbf{y}$. $P_1$ then performs  element-wise multiplications (step(b) in Figure~\ref{Fig:Matrix-vector multiplication}) for $\mathbf{X}$ and $\mathbf{Y}$ under the ciphertext. To construct the additive sharing of $\mathbf{z}= \mathbf{X} \ast \mathbf{y}$, $P_1$ uniformly chooses a random matrix $\mathbf{R}\in \mathbb{F}_p^{3\times 4}$ and computes $\langle \mathbf{Z} \rangle_0=\mathbf{X} \odot \mathbf{Y}-\mathbf{R}$ (step(c) in Figure~\ref{Fig:Matrix-vector multiplication}). $P_1$ sends the ciphertext result to $P_0$. $P_0$ decrypts it and sums each row in plaintext to obtain vector $\langle \mathbf{z} \rangle_0$ (step(d) in Figure~\ref{Fig:Matrix-vector multiplication}), similarly, $P_1$ performs the same operation on matrix $\mathbf{R}$ to obtain $\langle \mathbf{z} \rangle_1$.

\textit{Remark 4.1}: Compared to generating multiplication triples  for single multiplication \cite{keller2018overdrive,keller2020mp}, our constructed matrix-multiplication triples enable the communication overhead to be independent of the number of  multiplications, only related to the size of the input. This reduces the amount of data that needs to be exchanged between $P_0$ and $P_1$. In addition, we move the majority of the computation to be executed by the  semi-honest party, which avoids the need for distributed decryption and frequent zero-knowledge proofs in malicious adversary settings. Compared to existing parallel homomorphic computation methods \cite{jiang2018secure,halevi2014algorithms}, our matrix-vector multiplication does not involve any rotation operation, which is very computationally expensive compared to other homomorphic operations. This stems from our observation of the inner tie between HE and secret sharing. Since the final ciphertext result needs to be secretly shared to $P_0$ and $P_1$, we can first perform the secret sharing under the ciphertext (see step(c) and step(d) in Figure~\ref{Fig:Matrix-vector multiplication}), and then perform all rotation and summation operations under the plaintext.

\noindent\textbf{Security}. Our protocol for generating matrix-vector multiplication triples, $\pi_{Mtriple}$, is secure against the malicious model holder $P_0$  and the semi-honest client $P_1$. We provide the following theorem and prove it in Appendix~\ref{Proof of Theorem 1}.
\begin{theorem}
\label{theorem1}
Let the fully homomorphic encryption used in $\pi_{Mtriple}$  have the properties defined in Section~\ref{Fully Homomorphic Encryption}. $\pi_{Mtriple}$ is secure against the malicious model holder $P_0$  and the semi-honest client $P_1$.
\end{theorem}
%\begin{proof}
%Please refer to Appendix~\ref{Proof of Theorem 1}.
%\end{proof}

\subsubsection{\quad \; Generating convolution triple}
\label{Generating convolution triple}

We describe the technical details of generating authenticated  triples for convolution. Briefly, for a given convolution operation, we first convert it to equivalent matrix multiplications, and then generate triples for the matrix multiplications. We start by reviewing the definition of convolution and how to translate it into the equivalent matrix multiplication. Then, we explain how to generate authenticated triples.

\noindent\textcircled{1}\textbf{ Convolution}. Assuming an input tensor of size $u_w\times u_h$ with $c_i$ channels, denoted as $\mathbf{X}_{ijk}$, where $1\leq i\leq u_w$ and $1\leq j\leq u_h$ are spatial coordinates, and  $1\leq k\leq c_i$ is the channel. Let $c_o$  kernels with  a size of $(2l+1)\times (2l+1)\times c_i$ denote as tensor $\mathbf{Y}_{\Delta_i, \Delta_j, k, k'}$, where $-l\leq \Delta_i, \Delta_j \leq l$ are shifts of the spatial coordinates, $1\leq k\leq c_i$ and $1\leq k'\leq c_o$ are the channels and kernels, respectively.  The convolution between $\mathbf{X}$ and $\mathbf{Y}$  (i.e., $\mathbf{Z}=\mathtt{Conv}(\mathbf{X}, \mathbf{Y})$) is defined as below:

\begin{small}
 \begin{equation}
 \label{eq7}
\begin{split}
\mathbf{Z}_{ijk'}=\sum_{\Delta_i, \Delta_j,k} \mathbf{X}_{i+\Delta_i, j+\Delta_j, k}\cdot \mathbf{Y}_{\Delta_i, \Delta_j, k'}
\end{split}
\end{equation}
\end{small}
The resulting tensor $\mathbf{Z}_{ijk'}$ has $u_w'\times u_h'$ spatial coordinates and $c_0$ channels. We have $u_w'=(u_w-(2l+1)+2p)/s+1$  and $u_h'=(u_h-(2l+1)+2p)/s+1$, where $p$  represents the number of turns to zero-pad the input, and $s$ represents the stride size of the kernel movement \cite{lecun2015deep}. Note that  the entries of  $\mathbf{X}$ to be zero if $i+\Delta_i$ or $j+\Delta_j$ are outside of the ranges $[1; u_w']$ and $[1;u_h']$, respectively.

\noindent\textcircled{2}\textbf{ Conversion between convolution and matrix multiplication}. Based on Eqn.(\ref{eq7}), we can easily convert convolution into an equivalent matrix multiplication. Specifically, we construct a matrix $\mathbf{X'}$ with dimension $u_w' u_h'\times (2l+1)^2\cdot c_i$, where $\mathbf{X'}_{(i, j)(\Delta_i, \Delta_j,k)}=\mathbf{X}_{i+\Delta_i, j+\Delta_j, k}$. Similarly, we construct a matrix $\mathbf{Y'}$ of dimension $(2l+1)^2\cdot c_i \times c_o$ such that $\mathbf{Y'}(\Delta_i, \Delta_j, k)k'=\mathbf{Y}_{\Delta_i, \Delta_j, k'}$. Then, the original convolution operation is transformed into $\mathbf{Z'}= \mathbf{X'}\ast \mathbf{Y'}$, where $\mathbf{Z'}_{(ij)k'}$ $=\mathbf{Z}_{ijk'}$. In Appendix~\ref{ Conversion between convolution and matrix multiplication}, we provide a detailed example to implement the above transformation.
\begin{table}[htb]
\centering
\small
\begin{tabular}{|p{8.0cm}|}
\hline \\
 \textbf{Input:} $\{P_b\}_{b\in \{0, 1\}}$ holds $\langle \mathbf{X}\rangle_b$  uniformly chosen from ${\mathbb{F}_p}^{u_w\times u_h\times c_i}$, and $ \langle \mathbf{Y}\rangle_b$ uniformly chosen from ${\mathbb{F}_p}^{(2l+1)\times(2l+1)\times c_i\times c_o}$.   In addition, $p_1$ holds a MAC key $\alpha$ uniformly chosen from $\mathbb{F}_p$\\
 \textbf{Output:} $P_b$ obtains $\{[\![\mathbf{X}]\!]_b, [\![\mathbf{Y}]\!]_b, [\![\mathbf{Z}]\!]_b\}_{b\in \{0, 1\}}$, where $\mathbf{Z}=\mathtt{Conv}(\mathbf{X}, \mathbf{Y})$.\\
 \textbf{Procedure}:
 \begin{packeditemize}
    \item[1.]  $P_0$ and  $P_1$ participate in a secure two-party computation such that $P_0$ obtains an FHE public-secret key pair ($pk$, $sk$) while $P_1$ obtains the public key $pk$. This process is performed only once.
    \item[2.] $P_0$ first converts $\langle \mathbf{X}\rangle_0 $  and $\langle \mathbf{Y}\rangle_0 $ into  equivalent matrixes $\langle \mathbf{X'}\rangle_0 $ and and $\langle \mathbf{Y'}\rangle_0 $, where $\langle \mathbf{X'}\rangle_0 \in {\mathbb{F}_p}^{u_w' u_h'\times (2l+1)^2\cdot c_i}$ while $\langle \mathbf{Y'}\rangle_0 \in {\mathbb{F}_p}^{(2l+1)^2\cdot c_i \times c_o}$.  Then,
    $P_0$ sends the encryptions $c_1\leftarrow \mathtt{Enc}(pk, \langle \mathbf{X'}\rangle_0)$  and $c_2\leftarrow \mathtt{Enc}(pk, \langle \mathbf{Y'}\rangle_0)$ to $P_1$ along with zero-knowledge (ZK) proofs of plaintext knowledge of the two ciphertexts.
    \end{packeditemize}
 \begin{packeditemize}
    \item[3.] $P_1$  also converts $\langle \mathbf{X}\rangle_1 $  and $\langle \mathbf{Y}\rangle_1 $ into  equivalent matrixes $\langle \mathbf{X'}\rangle_1 $ and and $\langle \mathbf{Y'}\rangle_1 $.  Then it
     samples $(\langle \alpha \mathbf{X'}\rangle_1, \langle \alpha \mathbf{Y'}\rangle_1, \langle \alpha \mathbf{Z'}\rangle_1, \langle \mathbf{Z'} \rangle_1) $, and  computes  $c_3=\mathtt{Enc}_{pk}( \alpha(\langle \mathbf{X'}\rangle_1+\langle \mathbf{X'}\rangle_0)-\langle \alpha \mathbf{X'}\rangle_1)$, $c_4=\mathtt{Enc}_{pk}( \alpha(\langle \mathbf{Y'}\rangle_1+\langle \mathbf{Y'}\rangle_0)-\langle \alpha \mathbf{Y}\rangle_1)$, $c_5= \alpha \boxtimes(\mathbf{c}_\mathbf{X'}\circledast \mathbf{c}_\mathbf{Y'})-\mathtt{Enc}_{pk}(\langle \alpha \mathbf{Z'}\rangle_1)$, and $c_6=(\mathbf{c}_\mathbf{X'}\circledast \mathbf{c}_\mathbf{Y'})-\mathtt{Enc}_{pk}(\langle \mathbf{Z}\rangle_1)$. $P_1$ sends $c_3$, $c_4$,  $c_5$ and   $c_6$ to $P_0$.
    \item[4.] $P_0$ decrypts $c_3$, $c_4$,  $c_5$ and   $c_6$ to obtain $(\langle \alpha \mathbf{X'}\rangle_0, \langle \alpha \mathbf{Y'}\rangle_0, \langle \alpha \mathbf{Z'}\rangle_0, \langle \mathbf{Z'} \rangle_0)$, respectively. Then,  Both $P_0$ and $P_1$ converts these matrices back into tensors to get $(\langle \alpha \mathbf{X}\rangle_b, \langle \alpha \mathbf{Y}\rangle_b, \langle \alpha \mathbf{Z}\rangle_b, \langle \mathbf{Z} \rangle_b)$ for $b=\{0, 1\}$.
    \item[5.] $P_b$ outputs $\{[\![\mathbf{X}]\!]_b, [\![\mathbf{Y}]\!]_b, [\![\mathbf{Z}]\!]_b\}_{b\in \{0, 1\}}$, where $\mathbf{Z}=\mathtt{Conv}(\mathbf{X}, \mathbf{Y})$.
    \end{packeditemize}\\
\hline
\end{tabular}
\caption{Algorithm $\pi_{Ctriple}$ for generating authenticated convolution triple}
\label{Algorithm of generating convolution triple}
%\vspace{-15pt}
\end{table}

\noindent\textcircled{3}\textbf{Generating convolution triple}. Figure~\ref{Algorithm of generating convolution triple} depicts the interaction between the model holder $P_0$ and the client $P_1$ to generate triples of convolution. Succinctly, $P_0$ first uniformly selects $\langle \mathbf{X'}\rangle_0$ and $ \langle \mathbf{Y'}\rangle_0 $ and sends their encryption to $P_1$, along with  zero-knowledge proofs about these ciphertexts (step 2 in Figure~\ref{Algorithm of generating convolution triple}). $P_1$ recovers $\mathbf{X'}$ and $\mathbf{Y'}$ under the ciphertext and then computes $(\langle \alpha \mathbf{X'}\rangle_0, \langle \alpha \mathbf{Y'}\rangle_0, \langle \alpha \mathbf{Z'}\rangle_0, \langle \mathbf{Z'} \rangle_0)$ (step 3 in Figure~\ref{Algorithm of generating convolution triple}). Then it returns the corresponding ciphertexts to $P_0$. $P_0$ decrypts these ciphertexts and computes $\langle \alpha \mathbf{X}\rangle_0, \langle \alpha \mathbf{Y}\rangle_0$, $\langle \alpha \mathbf{Z}\rangle_0$ and $\langle \mathbf{Z} \rangle_0$ (step 4 in Figure~\ref{Algorithm of generating convolution triple}). Finally, $P_b$ obtains $\{[\![\mathbf{X}]\!]_b, [\![\mathbf{Y}]\!]_b, [\![\mathbf{Z}]\!]_b\}_{b\in \{0, 1\}}$, where $\mathbf{Z}=\mathtt{Conv}(\mathbf{X}, \mathbf{Y})$.

\textit{Remark 4.2}: We utilize the method in \cite{jiang2018secure} to perform the homomorphic multiplication operations involved in generating convolution triples in parallel. Given the multiplication of two $d\times d$-dimensional matrices,
it reduces the computational complexity from $O(d^2)$ to $O(d)$, compared with the existing  method \cite{halevi2014algorithms}. Besides, \cite{jiang2018secure} requires only one ciphertext to represent a single matrix whereas existing work \cite{halevi2014algorithms} requires $d$ ciphertexts (assuming the number of plaintext slots $n$ in FHE  is greater than $d^2$). In addition,  compared to generating multiplication triples  for single multiplication \cite{keller2018overdrive,keller2020mp}, the communication overhead of our method is independent of the number of  multiplications, only related to the size of the input, i.e., reduce the communication cost  from cubic to quadratic (both
offline and online phases).

\textit{Remark 4.3}: We further exploit the properties of semi-honest clients to improve the performance of generating convolution triples. Specifically, for the multiplication of matrices $\mathbf{X}$ and $\mathbf{Y}$, the permutations $\sigma(\mathbf{X})$ and $\varphi(\mathbf{Y})$ can be done in plaintext beforehand, which reduces the rotation in half compared to the original method (see Section~3.2 in \cite{jiang2018secure} for comparison). Moreover, we move the majority of the computation to be executed by the  semi-honest party, which avoids the need for distributed decryption and frequent zero-knowledge proofs in malicious adversary settings.

\noindent\textbf{Security}. Our protocol for generating authenticated convolution triples, $\pi_{Ctriple}$, is secure against the malicious model holder $P_0$  and the semi-honest client $P_1$. We provide the following theorem.
\begin{theorem}
\label{theorem2}
Let the fully homomorphic encryption used in $\pi_{Ctriple}$  have the properties defined in Section~\ref{Fully Homomorphic Encryption}. $\pi_{Ctriple}$ is secure against the malicious model holder $P_0$  and the semi-honest client $P_1$.
\end{theorem}
\begin{proof}
The proof logic of this theorem is very similar to \textbf{Theorem}~\ref{theorem1},  we omit it for brevity.
\end{proof}

\subsubsection{\quad\; Preprocessing for the nonlinear layer}
\label{Preprocessing for the nonlinear layer}
This process is performed by the client to generate garbled circuits of nonlinear functions for the model holder. Note that we do not generate  $\mathtt{GC}$ for ReLu but for the nonlinear part of ReLU, i.e. $sign(\mathbf{v})$ given an arbitrary input $\mathbf{v}$. We first define a truncation function $\mathbf{Trun}_h: \{0, 1\}^\lambda\rightarrow \{0, 1\}^h$, which outputs the last $h$ bits of the input, where $\lambda$ satisfies $\lambda\geq 2\kappa$. Then, the client is required to generate random ciphertexts and send them to the model holder as follows.
\begin{packeditemize}
\item Given the security parameter $\lambda$, and the boolean circuit $booln^C$ denoted  the nonlinear part of ReLU,  $P_1$ computes $\mathtt{Garble}(1^\lambda, booln^C)\rightarrow (\mathtt{GC}, \{ \{\mathtt{lab}_{i,j}^{in}\}_{i\in[2\kappa]},$ $\{\mathtt{lab}_{i,j}^{out}\}_{i\in[2\kappa]}\}_{j\in\{0,1\}})$, where $\mathtt{GC}$ is the  garbled circuit of $booln^C$, $\{ \{\mathtt{lab}_{i,j}^{in}\}_{i\in[2\kappa]},$ $\{\mathtt{lab}_{i,j}^{out}\}_{i\in[2\kappa]}\}_{j\in\{0,1\}}$ represent all possible garbled input and output labels, respectively. $P_1$ sends $\mathtt{GC}$ to the model holder $P_0$.
\item $P_1$ uniformly selects $\eta_{i, 1}$, $\gamma_{i, 1}$ and $\iota_{i, 1}$ from $\mathbb{F}_p$ for every $i\in[\kappa]$. Then, $P_1$ sets   $(\eta_{i, 0}, \gamma_{i, 0}, \iota_{i, 0})=(1+\eta_{i, 1}, \alpha+\gamma_{i, 1}, \alpha+\iota_{i, 1})$.
\item $P_1$ parses $\{\mathtt{lab}_{i,j}^{out}\}$ as $\varsigma_{i, j}||\vartheta_{i,j}$ for  every $i\in [2\kappa]$ and $j\in \{0, 1\}$, where $\varsigma_{i, j}\in\{0, 1\}$ and $\vartheta_{i,j}\in\{0,1\}^{\lambda-1}$.
\item For every $i\in [\kappa]$ and $j\in \{0, 1\}$, $P_1$ sends $ct_{i, \varsigma_{i, j}}$ and $\hat{ct}_{i, \varsigma_{i+\kappa, j}}$ to $P_0$, where $ct_{i, \varsigma_{i, j}}=\iota_{i, j}\oplus \mathbf{Trun}_\kappa(\vartheta_{i,j})$ and $\hat{ct}_{i, \varsigma_{i+\kappa, j}}=(\eta_{i, j}||\gamma_{i, j})\oplus \mathbf{Trun}_{2\kappa}(\vartheta_{i+\kappa,j})$.
\end{packeditemize}

\noindent\textbf{Security}. We leave the explanation of above ciphertexts sent by $P_1$ to $P_0$ to the following sections. Here we briefly describe the security of preprocessing for nonlinear layers. It is easy to infer that the above preprocessing for the nonlinear layer is secure against the semi-honest client $P_1$ and the malicious model holder $P_0$. Specifically, for the client $P_1$, since the entire preprocessing process does not require the participation of the model holder, the client cannot obtain any private information about the model holder. Similarly, for the malicious model holder $P_0$, since the preprocessing is non-interactive and the generated ciphertext satisfies the $\mathtt{GC}$ security defined in Section~\ref{garbled circuits}, $P_0$ cannot obtain the plaintext corresponding to the ciphertext sent by the client.

\subsection{Online Phase}
\label{Online Phase}
In this section,  we describe the online phase of \Name. We first explain how \Name utilizes the triples generated in the offline phase to generate authenticated shares for matrix-vector multiplication and convolution. Then, we describe the technical details of the nonlinear operation.

\subsubsection{\quad\; Perform linear layers in the online phase}
\label{Perform linear layers in the online phase}
\begin{table}[htb]
\centering
\small
\begin{tabular}{|p{8.0cm}|}
\hline
\textbf{Preamble}: Consider a neural network (NN) consists of $m$  linear layers and $m-1$  nonlinear layers.  Let the specification of the linear layer is $\mathbf{L}_1, \mathbf{L}_1, \cdots \mathbf{L}_{m}$ and the non-linear layer is $f_1, \cdots, f_{m-1}$.\\
 \textbf{Input:}$P_0$ holds $\{ \mathbf{L}_i\}_{i\in[m]}$, \textit{i.e.}, weights for the $m$ linear layers. $P_1$  holds $\mathbf{x}_0$ as the input of NN, a random MAC key $\alpha$ from $\mathbb{F}_p$ to be used throughout the protocol execution.\\
 \textbf{Output:} $P_b$ obtains $[\![\mathbf{v}_i=\mathbf{L}_i \mathbf{x}_{i-1}]\!]_b$ for $i\in[m]$ and $b=\{0, 1\}$.\\
 \textbf{Procedure}:\\
\textbf{Input Sharing}:
\begin{packeditemize}
\item[1.] To share $P_0$'s input $\{ \mathbf{L}_i\}_{i\in[m]}$, all parties pick up a fresh authenticated element $[\![\mathbf{R}_i]\!]$ of the same dimension as $\mathbf{L}_i$.
\item[2.] $[\![\mathbf{R}_i]\!]$ is opened to $P_0$, and then it sends $\varpi_i=\mathbf{L}_i-\mathbf{R}_i$ to $P_1$.
\item[3.]  $P_b$ locally computes $[\![\mathbf{L}_i]\!]_b=[\![\mathbf{R}_i]\!]_b+\varpi_i$ for $b=\{0, 1\}$.
\item[4.] To share $P_1$'s input $\mathbf{v}_0$, $P_1$ randomly selects two masks $\xi$ and $\zeta$ of the same dimension as $\mathbf{v}_0$. Then, it sends  $[\![\mathbf{v}_0]\!]_0=(\mathbf{v}_0-\xi, \alpha\mathbf{v}_0-\zeta)$ to $P_0$. $P_1$ sets $[\![\mathbf{v}_0]\!]_1=(\xi, \zeta)$.
\item[5.] For each $i\in[m]$,
\begin{packeditemize}
\item  \textbf{Matrix-vector Multiplication}:  To  generate an authenticated triple of multiplications between matrix $\mathbf{A}$ and vector $\mathbf{b}$, where $\mathbf{A}$ and $\mathbf{b}$ are variables generated in the inference process. $P_0$ and $P_1$ take a fresh authenticated matrix-vector triple $\{[\![\mathbf{X}]\!]_b, [\![\mathbf{y}]\!]_b, [\![\mathbf{z}]\!]_b\}_{b\in \{0, 1\}}$ of dimensions consistent with $\mathbf{A}$ and $\mathbf{b}$.  Then, both party open $\mathbf{A}-\mathbf{X}$ and $\mathbf{b}-\mathbf{y}$. Finally, $P_b$ locally computes $[\![\mathbf{A}\ast \mathbf{b}]\!]_b$ based on Eqn.(\ref{eq6}).
\item \textbf{Convolution}: To  generate an authenticated triple of Convolution between tensors $\mathbf{A}$ and $\mathbf{B}$, where $\mathbf{A}$ and $\mathbf{B}$ are variables generated in the inference process. $P_0$ and $P_1$ take a fresh authenticated Convolution triple $\{[\![\mathbf{X}]\!]_b, [\![\mathbf{Y}]\!]_b, [\![\mathbf{Z}]\!]_b\}_{b\in \{0, 1\}}$ of dimensions consistent with $\mathbf{A}$ and $\mathbf{B}$.  Then, both party open $\mathbf{A}-\mathbf{X}$ and $\mathbf{b}-\mathbf{Y}$. Finally, $P_b$ locally computes $[\![ \mathtt{Conv}(\mathbf{A}, \mathbf{B})]\!]_b$ based on Eqn.(\ref{eq6}).
\end{packeditemize}
\item[6.] $P_b$ obtains $[\![\mathbf{v}_i=\mathbf{L}_i \mathbf{x}_{i-1}]\!]_b$ for $i\in[m]$ and $b=\{0, 1\}$.
\end{packeditemize}\\
\hline
\end{tabular}
\caption{Online linear layers protocol $\pi_{OLin}$}
\label{Online linear layers protocol}
%\vspace{-10pt}
\end{table}
Figure~\ref{Online linear layers protocol} depicts the interaction of the model holder and the client to perform linear layer operations in the online phase. Specifically, given the model holder's input $\{ \mathbf{L}_i\}_{i\in[m]}$  and the client's input $\mathbf{v}_0$, both parties first generate  authenticated shares  of their respective inputs (steps 1-4 in Figure~\ref{Online linear layers protocol}). Since the client is considered semi-honest, its input is shared more efficiently than the model holder, i.e. only local computations are required on randomly selected masks,  while the sharing process of model holder's input is consistent with the previous  malicious settings \cite{keller2018overdrive,keller2020mp,chen2020maliciously}. After that, the model holder and the client use the triples generated in the offline phase (i.e., matrix-vector multiplication triples and convolution triples) to generate authenticated sharing of linear layer computation results (step 5 in Figure~\ref{Online linear layers protocol}).

\noindent\textbf{Security}. Our protocol for  performing  linear layer operations in the online phase, $\pi_{OLin}$, is secure against the malicious model holder $P_0$  and the semi-honest client $P_1$. We provide the following theorem.
\begin{theorem}
\label{theorem2}
Let triples  used in $\pi_{OLin}$  are generated from $\pi_{Mtriple}$ and $\pi_{Ctriple}$. $\pi_{OLin}$ is secure against the malicious model holder $P_0$  and the semi-honest client $P_1$.
\end{theorem}
\begin{proof}
The proof logic of this theorem is identical to that of \cite{damgaard2012multiparty}.  Interested readers can refer to \cite{damgaard2012multiparty} for more details.
\end{proof}

\subsubsection{\quad \; Perform non-linear layers in the online phase}
\label{Perform non-linear layers in the online phase}

\begin{table}[htb]
\centering
\small
\begin{tabular}{|p{8.0cm}|}
\hline
 \textbf{Input:}$P_0$ holds $[\![\mathbf{v}_i]\!]_0$  and $P_1$ holds $[\![\mathbf{v}_i]\!]_1$ for $i\in[m]$ and $b=\{0, 1\}$. In addition, $P_1$ holds the MAC key $\alpha$.\\
 \textbf{Output:} $P_b$ obtains $[\![\mathbf{x}_i=ReLU(\mathbf{v}_i)]\!]_b$  and $\langle \alpha\mathbf{v}_i\rangle_b$ for $i\in[m]$ and $b=\{0, 1\}$.\\
 \textbf{Procedure}(take single $\mathbf{v}_i$ as an example):
 \begin{packeditemize}
\item[1.] Garbled Circuit Phase:
\begin{packeditemize}
\item $P_0$ and $P_1$ invoke  the OT$_{\lambda}^{\kappa}$ (see Section~\ref{Oblivious Transfer}), where $P_1$'s inputs are  $\{\mathtt{lab}_{j,0}^{in}, \mathtt{lab}_{j,1}^{in}\}_{j\in\{\kappa+1, \cdots, 2\kappa\}}$ while $P_0$'s input is $\left \langle \mathbf{v}_i \right \rangle_0$. Hence, $P_0$ learns  $\{\mathtt{\tilde{lab}}_{j}^{in}\}_{j\in\{\kappa+1, \cdots, 2\kappa\}}$. Also, $P_1$ sends its garbled inputs $\{ \{\mathtt{\tilde{lab}}_{j}^{in}=\mathtt{lab}_{j, \left \langle \mathbf{v}_{i} \right \rangle_1[j]}\}_{j\in[\kappa]}$ to $P_0$.
\item With $\mathtt{GC}$  and  $\{\mathtt{\tilde{lab}}_{j}^{in}\}_{j\in[2\kappa]}$, $P_0$ evaluates $\mathtt{GCEval}(\mathtt{GC}, \{\mathtt{\tilde{lab}}_{j}^{in}\}_{j\in[2\kappa]})\rightarrow \{\mathtt{\tilde{lab}}_{j}^{out}\}_{j\in[2\kappa]}$.
\end{packeditemize}
\item[2.] Authentication Phase 1:
\begin{packeditemize}
\item  $P_0$ parses $\mathtt{\tilde{lab}}_{j}^{out}$ as $\tilde{\varsigma}_{j}||\tilde{\vartheta}_{j}$ where $\tilde{\varsigma}_{j}\in\{0, 1\}$ and $\tilde{\vartheta}_{j}\in\{0,1\}^{\lambda-1}$ for every $j\in[2\kappa]$.
\item  $P_0$ computes $c_j=ct_{j, \tilde{\varsigma}_{j}}\oplus \mathbf{Trun}_\kappa(\tilde{\vartheta}_{j})$ and $(d_j||e_j)=\hat{ct}_{j, \tilde{\varsigma}_{i+\kappa}}\oplus\mathbf{Trun}_{2\kappa}(\tilde{\vartheta}_{j+\kappa})$ for every $j\in[\kappa]$.
\end{packeditemize}
\item[3.] Local Computation Phase:
\begin{packeditemize}
\item $P_1$ outputs $\langle g_1\rangle_1=(-\sum_{j\in[\kappa]}\iota_{j, 1}2^{j-1})$, $\langle g_2\rangle_1=(-\sum_{j\in[\kappa]}\eta_{j, 1}2^{j-1})$ and $\langle g_3\rangle_1=(-\sum_{j\in[\kappa]}\gamma_{j, 1}2^{j-1})$.
\item $P_0$ outputs $\langle g_1\rangle_0=(\sum_{j\in[\kappa]}c_j2^{j-1})$, $\langle g_2\rangle_0=(\sum_{j\in[\kappa]}d_j2^{j-1})$ and $\langle g_3\rangle_0=(\sum_{j\in[\kappa]}e_j2^{j-1})$.
\end{packeditemize}
\item[4.] Authentication Phase 2:
\begin{packeditemize}
\item For every $\mathbf{v}_i$ where $i\in [m]$, $P_b$ randomly select a fresh authenticated triple $\{[\![x]\!]_b, [\![y]\!]_b, [\![z]\!]_b\}_{b\in \{0, 1\}}$.
\item All parties reveal $\mathbf{v}_i-x$ and $g_2-y$ to each other, and then locally compute $\langle z_2 \rangle_b=\langle \mathbf{v}_i\cdot sign(\mathbf{v}_i) \rangle_b$ and $\langle z_3 \rangle_b=\langle \alpha \mathbf{v}_i\cdot sign(\mathbf{v}_i) \rangle_b$ based on Eqn.(\ref{eq6}).
\item $P_b$ obtains $[\![\mathbf{x}_i=ReLU(\mathbf{v}_i)]\!]_b=(\langle z_2 \rangle_b, \langle z_3 \rangle_b)$  and $\langle \alpha\mathbf{v}_i\rangle_b=\langle g_1\rangle_b$.
\end{packeditemize}
\end{packeditemize}\\
\hline
\end{tabular}
\caption{Online non-linear layers protocol $\pi_{ONlin}$}
\label{Online non-linear layers protocol}
%\vspace{-15pt}
\end{table}

In this section, we present the technical details of the execution of nonlinear functions in the online phase. We mainly focus on how to securely compute the activation function ReLU, which is the most representative nonlinear function in deep neural networks. As shown in Figure~\ref{Online linear layers protocol},  the result $\mathbf{v}_i$ obtained from each linear layer $\mathbf{L}_i$ is held by both parties in the format of authenticated  sharing. Similarly, for the function $f_i$ in the $i$-th nonlinear layer, the goal of \Name is to securely compute $f_i(\mathbf{v}_i)$ and share it  to the model holder and client in the authenticated sharing manner.  We describe  details in  \textbf{Figure}~\ref{Online non-linear layers protocol}.

%In the following, we describe the details and illustrate our optimization method in designing the non-linear layer protocol.

%  Our protocol can be divided into four phases: Garbled Circuit phase, Authentication phase 1, Local Computation phase, and Authentication phase 2.  To compute the ReLU function in the $i$-th nonlinear layer, given $P_0$'s input $[\![\mathbf{v}_i]\!]_0$  and $P_1$'s input  ($[\![\mathbf{v}_i]\!]_1$, $\alpha$), \textbf{Figure}~\ref{Online non-linear layers protocol} gives the detailed technical description.

\textbf{Garbled Circuit Phase}. As described in Section~\ref{Preprocessing for the nonlinear layer}, in the offline phase, $P_1$ constructs a $\mathtt{GC}$ for the nonlinear part of ReLU (i.e., $sign(\mathbf{v_i})$  for arbitrary input $\mathbf{v_i}\in \mathbb{F}_p$) and sent it to $P_0$. In the online phase,  $P_0$ and $P_1$ invoke  the OT$_{\lambda}^{\kappa}$, where $P_1$ as the sender whose inputs are  $\{\mathtt{lab}_{j,0}^{in}, \mathtt{lab}_{j,1}^{in}\}_{j\in\{\kappa+1, \cdots, 2\kappa\}}$ while $P_0$'s (as the receiver) input is $\left \langle \mathbf{v}_i \right \rangle_0$. As a result, $P_0$ gets set of garbled inputs of $\mathbf{v_i}$ in $\mathtt{GC}$. Then, $P_0$ evaluates $\mathtt{GC}$ with garbled inputs of $\mathbf{v_i}$ and learns  the set of output labels for the bits of $\mathbf{v}_i$ and $sign(\mathbf{v}_i)$.

\textbf{Authentication Phase 1}. This phase aims to calculate the share of the authentication of each bit of $\mathbf{v_i}$, i.e., $sign(\mathbf{v}_i)[j], \alpha sign(\mathbf{v}_i)[j]$,   and $\alpha \mathbf{v}_i[j]$  for $j\in[\kappa]$, based on the previous phase. We take an example of how to calculate $\alpha \mathbf{v}_i$. It is clear that the share of $\alpha \mathbf{v}_i[j]$ is either 0 or $\alpha$ depending on whether $\mathbf{v}_i[j]$ is 0 or 1. Recall that the output of the GC is two output labels corresponding to each $\mathbf{v}_i[j]$ (each one for $\mathbf{v}_i[j]=0$ and 1). We use the symbol $\mathtt{lab}_{j,0}^{out}$ and $\mathtt{lab}_{j,1}^{out}$ to denote $\mathbf{v}_i[j]=0$ and $\mathbf{v}_i[j]=1$, respectively.  To calculate the shares of $\alpha\mathbf{v}_i[j]$, $P_1$ randomly selects $\iota_{j}\in \mathbb{F}_p$ in the offline phase and encrypts it as  $\mathtt{lab}_{j,1}^{out}$ and encrypts $\iota_{j}+\alpha$  as $\mathtt{lab}_{j,0}^{out}$.  $P_1$ sends the two ciphertexts to $P_0$ and sets its own share of   $\alpha \mathbf{v}_i[j]$ to $-\iota_{j}$.  Since $P_0$ has obtained $\mathtt{lab}_{j,\mathbf{v}_i[j]}^{out}$  in the previous phase, it can definitely decrypt it and obtain its own share of $\alpha \mathbf{v}_i[j]$. Computation of $sign(\mathbf{v}_i)[j]$ and $ \alpha sign(\mathbf{v}_i)[j]$ follows a similar logic, utilizing the random values $\eta_{j, 1}$, $\gamma_{j, 1}$ sent by $P_1$ to $P_0$ in the offline phase, respectively.

\textbf{Local Computation Phase}. This process is used to calculate the share of $sign(\mathbf{v}_i), \alpha sign(\mathbf{v}_i)$,   and $\alpha \mathbf{v}_i$ based on the results learned by all parties in the previous stage. For example, to compute the share of $\alpha \mathbf{v}_i$, each party locally multiplies the share of $\alpha \mathbf{v}_i[j]$ with $2^{j-1}$ and sums all the resultant values. Each party computes the share of $sign(\mathbf{v}_i)$ and $ \alpha sign(\mathbf{v}_i)$ in a similar manner.

\textbf{Authentication Phase 2}. We compute  the shares of $ReLU(\mathbf{v}_i)=\mathbf{v}_isign(\mathbf{v}_i)$, and $\alpha ReLU(\mathbf{v}_i)$. Since each party holds the authenticated shares of $\mathbf{v}_i$ and $sign(\mathbf{v}_i)$, we can achieve this based on Eqn.(\ref{eq6}).

\textit{Remark 4.4}. We adopt two methods to minimize the number of multiplication operations involved in the $\mathtt{GC}$. One is to compute the  garbled output of per-bit of $sign(\mathbf{v}_i)$ in $\mathtt{GC}$. Another is to encapsulate only the nonlinear part of ReLU into $\mathtt{GC}$. In this way, we avoid  computing $\alpha ReLU(\mathbf{v}_i)$ and $ ReLU(\mathbf{v}_i)$ in $\mathtt{GC}$, which is multiply operation intensive. Compared with works \cite{keller2018overdrive,keller2020mp,chen2020maliciously} with  malicious adversary, \Name reduces the communication overhead of each ReLU function from $2c\lambda+190\kappa\lambda+232\kappa^2$ to $2d\lambda+4\kappa\lambda+6\kappa^2$, where $d\ll c$.

\textit{Remark 4.5}. We devise a lightweight method to check whether the model holder's input at the non-linear layer is consistent with what it has learned at the previous layer. Specifically, at the end of evaluating the $i-1$-th linear layer,  both parties learns the share of $\alpha \mathbf{v}_i$. Then, $\mathbf{v}_i$ is used as the input of the $i$-th nonlinear.  To check that $P_0$ is fed the correct input, We require $\alpha \mathbf{v}_i$  to be recomputed in $\mathtt{GC}$ and share  again to  both parties. Therefore, after evaluating each nonlinear layer, both parties hold two independent shares of $\alpha \mathbf{v}_i$. This provides a way to determine if $P_0$ provided the correct  input by verifying that  the two independent shares are consistent (See Section~\ref{Consistency check} for more details).

\vspace{5pt}
\noindent\textbf{Correctness}. We analyze the correctness of our protocol $\pi_{ONlin}$ as follows. Based on the correctness of OT$_{\lambda}^{\kappa}$, the model holder $P_0$ holds $ \{\mathtt{\tilde{lab}}_{j}^{in}=\mathtt{lab}_{j, \left \langle \mathbf{v}_i \right \rangle_0[j]}\}_{j\in\{\kappa+1, \cdots 2\kappa\}}$. Using  $ \{\mathtt{\tilde{lab}}_{j}^{in}$ $=\mathtt{lab}_{j, \left \langle \mathbf{v}_i \right \rangle_1[j]}\}_{j\in[\kappa]}$ for $j\in[\kappa]$,  and the correctness of ($\mathtt{Garble}$, $\mathtt{GCEval}$) for circuit $booln^f$,  we learn  $\mathtt{\tilde{lab}}_{j}^{out}=\mathtt{lab}_{j,\mathbf{v}_i[j]}^{out}$ and $\mathtt{\tilde{lab}}_{j+\kappa}^{out}=\mathtt{lab}_{j+\kappa,sign(\mathbf{v}_i)[j]}^{out}$, for $j\in[\kappa]$. Therefore, for $i\in[k]$,
we have $\tilde{\varsigma}_{j}||\tilde{\vartheta}_{j}=\varsigma_{j, \mathbf{v}_i[j]}||\vartheta_{j, \mathbf{v}_i[j]}$ and
$\tilde{\varsigma}_{j+\kappa}||\tilde{\vartheta}_{j+\kappa}=\varsigma_{j+\kappa, sign(\mathbf{v})_i[j]}||\vartheta_{j+\kappa, sign(\mathbf{v}_i)[j]}$. Hence, $c_j=ct_{j, \varsigma_{j, \mathbf{v}_i[j]}}$ $\oplus\mathbf{Trun}_{\kappa}(\vartheta_{j, \mathbf{v}_i[j]})=\iota_{j, \mathbf{v}_i[j]}$ and $(d_j||e_j)=\hat{ct}_{j, \varsigma_{j+\kappa, sign(\mathbf{v}_i)[j]}}\oplus\mathbf{Trun}_{2\kappa}(\vartheta_{j+\kappa, sign(\mathbf{v}_i)[j]})=\eta_{j, sign(\mathbf{v}_i)[j]}||\gamma_{j,  sign(\mathbf{v}_i)[j]}$. Based on these, we have
\begin{packeditemize}
\begin{small}
\item  $g_1=\sum_{j\in[\kappa]}(c_j-\iota_{j, 0})2^{j-1}=\sum_{j\in[\kappa]}\alpha(\mathbf{v}_i[j])2^{j-1}=\alpha\mathbf{v}_i$.
\item $g_2=\sum_{j\in[\kappa]}(d_j-\eta_{j, 0})2^{j-1}=\sum_{j\in[\kappa]}(sign(\mathbf{v}_i)[j])2^{j-1}=sign(\mathbf{v}_i)$.
 \item $g_3=\sum_{j\in[\kappa]}(e_j-\gamma_{j, 0})2^{j-1}=\sum_{j\in[\kappa]}\alpha (sign(\mathbf{v}_i)[j])2^{j-1}=\alpha sign(\mathbf{v}_i)$.
 \end{small}
\end{packeditemize}
Since each party  holds the authenticated shares of $\mathbf{v}_i$ and $sign(\mathbf{v}_i)$, we can easily compute the shares of $f(\mathbf{v}_i)=\mathbf{v}_isign(\mathbf{v}_i)$, and $\alpha f(\mathbf{v}_i)$. This concludes the correctness proof.

\noindent\textbf{Security}. Our protocol for  performing  nonlinear layer operations in the online phase, $\pi_{ONlin}$, is secure against the malicious model holder $P_0$  and the semi-honest client $P_1$. We provide the following theorem and prove it in Appendix~\ref{Proof of Theorem 4}.
\begin{theorem}
Let ($\mathtt{Garble}$, $\mathtt{GCEval}$) be a garbling scheme with the properties defined in Section~\ref{garbled circuits}. Authenticated shares have the properties defined in Section~\ref{Secret Sharing}. Then our  protocol $\pi_{ONlin}$ is secure against the malicious model holder $P_0$  and the semi-honest client $P_1$.
\end{theorem}
%\begin{proof}
%Please refer to Appendix~\ref{Proof of Theorem 4}.
%\end{proof}

\subsection{Consistency Check}
\label{Consistency check}
 \Name performs $\pi_{OLin}$ and $\pi_{ONlin}$ alternately in the online phase to output the inference result $\mathbf{M}(\mathbf{x}_0)$ for a given input $\mathbf{x}_0$, where all intermediate results output by the nonlinear layer and the linear layer are held on $P_0$ and $P_1$ in an authenticated sharing manner. To verify the correctness of  $\mathbf{M}(\mathbf{x}_0)$, the client needs to perform a consistency check on all computed results. If the verification passes, $P_1$ locally evaluates the fairness of the ML model based on Eqn.(\ref{eq2}). Otherwise, abort. In more detail, for sharing $P_0$'s input and executing  each linear layer $\{ \mathbf{L}_i\}_{i\in[m]}$,  \Name  needs to pick up a large number of fresh authenticated single elements or triples (see Figure~\ref{Online linear layers protocol}) and open them for computation. Assume that the set of all opened elements is $(a_1, a_2 \cdots, a_t)$, and $P_b$ holds $\left \langle \mathbf{\rho}_i \right \rangle_b=\left \langle \alpha a_i \right \rangle_b$ as well as $\left \langle \mathbf{\tau}_i \right \rangle_b=\left \langle a_i \right \rangle_b$, we need to perform a consistency check to verify  $\mathbf{\rho}_i-\alpha  \mathbf{\tau}_i=0$. Beside, For executing each nonlinear layer $\{f_i\}_{i\in[m-1]}$, the inputs of $\pi_{ONlin}$ are shares of  $\mathbf{v}_i$ and  $\mathbf{\tau}_i=\alpha \mathbf{v}_i$.  To check that $P_0$ is fed the correct input, We require $\alpha \mathbf{v}_i$  to be recomputed in the $\mathtt{GC}$ and share it again to  both parties, denoting the new  $\alpha \mathbf{v}_i$ as $\mathbf{\xi}_i$.
 We also need to perform a consistency check to verify  $\sum_{i=1}^{i=m} \mathbf{\tau}_i-\mathbf{\xi}_i=\mathbf{0}$.
  % Based on the above, the consistency check used  to verify the correctness of \Name are as follows.
% \begin{itemize}
%\item \textbf{Correctness of the model holder's initial input}: For $P_0$'s input $\{ \mathbf{L}_i\}_{i\in[m]}$, verify  that  $P_0$ shares the correct input by checking  $\mathbf{\zeta}_i=\mathbf{0}$, where $\left \langle \mathbf{\zeta}_i \right \rangle_b= \alpha_b \left \langle \mathbf{L}_i \right \rangle_b- \left \langle \alpha \mathbf{L}_i \right \rangle_b$.
% \textbf{Correctness of linear layers execution}: For each $i\in [m]$, verify the  correctness of executing $\pi_{OLin}$ by checking   $\mathbf{\rho}_i=\mathbf{0}$.
%\item \textbf{Correctness of nonlinear layers execution}: For each $i\in [m-1]$, verify the  correctness of executing  $\pi_{ONlin}$ by checking $\mathbf{\tau}_i-\mathbf{\xi}_i=\mathbf{0}$.
%\end{itemize}

\begin{table}[htb]
\centering
\small
\begin{tabular}{|p{8.0cm}|}
\hline
 \textbf{Input:}$P_b$ $b\in \{0, 1\}$ holds $\left \langle \tau_i \right \rangle _b$, $\left \langle \xi_i \right \rangle _b$ and $[\![a_j]\!]_b$ for $i\in[m-1]$ and $j\in [t]$. \\
 \textbf{Output:} $P_1$ obtains $\mathbf{M}(\mathbf{x}_0)$ if verification passes. Otherwise, abort. \\
 \textbf{Procedure}
\begin{packeditemize}
\item For $i\in[m]$ and $j\in [t]$, $P_1$  uniformly samples $\mathbf{r}_i$ and $\mathbf{r}_j$ and sends them to $P_0$.
\item  $P_0$ computes $\left \langle q \right \rangle_0=\sum_{j\in [t]} \mathbf{r}_j(\left \langle \mathbf{\rho}_j \right \rangle_0-\alpha_0a_j)+ \sum_{i\in [m-1]}\mathbf{r}_i(\left \langle \mathbf{\tau}_i \right \rangle_0- \left \langle \mathbf{\xi}_i \right \rangle_0)$, and sends $\left \langle q \right \rangle_0$ to $P_1$.
\item  $P_1$ computes $\left \langle q \right \rangle_1=\sum_{j\in [t]} \mathbf{r}_j(\left \langle \mathbf{\rho}_j \right \rangle_1-\alpha_1a_j)+ \sum_{i\in [m-1]}\mathbf{r}_i(\left \langle \mathbf{\tau}_i \right \rangle_1- \left \langle \mathbf{\xi}_i \right \rangle_1)$.
\item $P_1$ aborts if $\left \langle q \right \rangle_0+ \left \langle q \right \rangle_1 \neq 0 \mod p$. Else, $P_1$ locally evaluates the fairness of the ML model based on Eqn.(\ref{eq2}) by reconstructing  $\mathbf{M}(\mathbf{x}_0)$.
\end{packeditemize}\\
\hline
\end{tabular}
\caption{Consistency check protocol $\pi_{Ocheck}$}
\label{Online check protocol}
%\vspace{-10pt}
\end{table}
Figure~\ref{Online check protocol} presents the details of consistency check, where  we combine all the above checks into a single check by using random scalars picked by $P_1$. The correctness of  $\pi_{Ocheck}$  can be easily deduced by inspecting the implementation of the protocol. Specifically, By correctness of $\pi_{OLin}$, we have $\mathbf{\rho}_j-\alpha  \mathbf{\tau}_j=(\left \langle \mathbf{\rho}_j \right \rangle_0-\alpha_0a_j+\left \langle \mathbf{\rho}_j \right \rangle_1-\alpha_1a_j)=0$  for  every linear layer $\{ \mathbf{L}_j\}_{j\in[m]}$.  By correctness of $\pi_{ONlin}$, we have $\mathbf{\tau}_i-\mathbf{\xi}_i=(\left \langle \mathbf{\tau}_i \right \rangle_0- \left \langle \mathbf{\xi}_i \right \rangle_0)+(\left \langle \mathbf{\tau}_i \right \rangle_1- \left \langle \mathbf{\xi}_i \right \rangle_1)=0$ for all nonlinear layers. Hence, we have $\left \langle q \right \rangle_0+ \left \langle q \right \rangle_1 =\sum_{j\in [t]} \mathbf{r}_j(\mathbf{\rho}_j-\alpha  \mathbf{\tau}_j)+  \sum_{i\in [m-1]}\mathbf{r}_i(\mathbf{\tau}_i-\mathbf{\xi}_i)=0$.

\noindent\textbf{Security}. We demonstrate that the consistency check protocol $\pi_{Ocheck}$ have an overwhelming probability to abort if $P_0$ tampered with the input during execution. We provide the following theorem and prove it in Appendix~\ref{Proof of Theorem 5}.
\begin{theorem}
In real execution, if $P_0$ tampers with its input, then $P_1$ aborts with probability at least $1-1/p$.
\end{theorem}
%\begin{proof}
%Please refer to Appendix~\ref{Proof of Theorem 5}.
%\end{proof}

 \section{Performance Evaluation}
\label{sec:performance evaluation}
In this section, we conduct experiments to demonstrate the performance of \Name. Since there is no secure inference protocol specifically designed for the malicious model holder threat model, we choose the state-of-the-art generic MPC framework Overdrive \cite{keller2018overdrive}\footnote{Although work \cite{chen2020maliciously} shows better performance compared to Overdrive, it is difficult to compare with \cite{chen2020maliciously} because of the unavailability of its code. However, we clearly outperform \cite{chen2020maliciously} by constructing a more efficient method to generate triples. In addition, \cite{chen2020maliciously} requires fitting nonlinear functions such as ReLU to a quadratic polynomial to facilitate computation, which is also contrary to the motivation of this paper.} as the baseline. Note that we also consider the client as a semi-honest entity when implementing Overdrive, so that Overdrive can also utilize the properties of semi-honest client  to avoid redundant verification and zero-knowledge proof. In this way, we can ``purely" discuss the technical advantages of \Name over Overdrive, while excluding the inherent advantages of \Name due to the weaker  threat model. Specifically, we analyze the performance of \Name from offline and online phases, respectively, where we discuss the superiority  of \Name over Overdrive in terms of computation and communication cost in performing linear and non-linear layers. In the end, We demonstrate the cost superiority of  \Name compared to Overdrive on mainstream models including ResNet-18 and LeNet.

\subsection{Implementation details}
 \Name is implemented through the C++ language and provides 128 bits of computational security and 40 bits of statistical security. The entire system operates on the 44-bit prime field. We utilize the SEAL homomorphic encryption library \cite{sealcrypto} to perform nonlinear layers including generative matrix-vector multiplication and convolution triples, where we set  the maximum number of slots allowed for a single ciphertext as 4096.  The garbled circuit for the nonlinear layer is constructed on the EMP toolkit \cite{wang2016emp} (with the OT protocol that resists active adversaries). Zero-knowledge proofs of plaintext knowledge are implemented based on MUSE \cite{lehmkuhl2021muse}. Our experiments are carried out in both the LAN and WAN settings. LAN is implemented with two workstations in our lab. The client workstation has AMD EPYC 7282 1.4GHz CPUs with 32 threads on 16 cores and 32GB RAM. The server workstation has Intel(R) Xeon(R) E5-2697 v3 2.6GHz CPUs with 28 threads on 14 cores and 64GB RAM. The WAN setting is based on a connection between a local PC and an Amazon AWS server with an average bandwidth of 963Mbps and running time of around 14ms.

\subsection{Performance of offline phase}
\label{Performance of offline phase}
%In the offline phase, the main cost of \Name is to generate triples. We  discuss the performance advantage of  \Name over  Overdrive in generating triples.
\subsubsection{\quad \; Cost of generating matrix-vector multiplication triple}
 \renewcommand\tablename{TABLE}
\renewcommand \thetable{\Roman{table}}
\setcounter{table}{0}
\setcounter{figure}{5}
%\vspace{-10pt}
\begin{table}[htb]
\centering
\small
\vspace{-10pt}
\caption{Cost of generating the matrix-vector multiplication triple}
\label{Complexity of Matrix-Vector Multiplication}
\setlength{\tabcolsep}{1.9mm}{
\resizebox{\linewidth}{!}{\begin{tabular}{|c|c|c|c|c|c|c|}
\Xhline{1pt}
\multirow{3}*{\textbf{Dimension}}&\multicolumn{2}{c|}{\textbf{Comm.cost (MB)}} &\multicolumn{4}{c|}{\textbf{Running time (s)} }\\ \cline{2-7}
&\multirow{2}*{Overdrive}&\multirow{2}*{\textbf{\Name}(Reduction)} & \multicolumn{2}{c|}{Overdrive} & \multicolumn{2}{c|}{ \Name (speedup)} \\
& & &  LAN & WAN & LAN & WAN\\
\Xhline{1pt}
{$1\times 4096$}& 27.1& {\makecell[c]{\textbf{2.1}\\\textbf{(12.9$\times$)}}}&{2.3}&{17.7}&{\makecell[c]{\textbf{0.9}\\\textbf{ (2.6$\times$)}}}&{\makecell[c]{\textbf{12.4}\\\textbf{(1.5$\times$)}}}\\ \hline
{$16\times 2048$}& 216.4& {\makecell[c]{\textbf{17.6}\\\textbf{(12.3$\times$)}}}&{15.3}&{26.2}&{\makecell[c]{\textbf{7.6}\\\textbf{ (2.0$\times$)}}}&{\makecell[c]{\textbf{14.1}\\\textbf{(1.6$\times$)}}}\\ \hline
{$16\times 4096$}& 432.8& {\makecell[c]{\textbf{34.5}\\\textbf{(12.5$\times$)}}}&{30.6}&{43.4}&{\makecell[c]{\textbf{15.1}\\\textbf{ (2.0$\times$)}}}&{\makecell[c]{\textbf{26.9}\\\textbf{(1.6$\times$)}}}\\ \hline
{$64\times 2048$}&865.6 & {\makecell[c]{\textbf{68.3}\\\textbf{(12.7$\times$)}}}&{60.9}&{72.4}&{\makecell[c]{\textbf{29.2}\\\textbf{ (2.1$\times$)}}}&{\makecell[c]{\textbf{40.7}\\\textbf{(1.7$\times$)}}}\\  \hline
{$64\times 4096$}& 1326.2& {\makecell[c]{\textbf{135.7}\\\textbf{(9.8$\times$)}}}&{103.0}&{114.8}&{\makecell[c]{\textbf{57.8}\\\textbf{ (1.8$\times$)}}}&{\makecell[c]{\textbf{68.2}\\\textbf{(1.6$\times$)}}}\\ \hline
{$128\times 4096$}& 2247.4& {\makecell[c]{\textbf{271.9}\\\textbf{(8.3$\times$)}}}&{187.1}&{199.1}&{\makecell[c]{\textbf{117.3}\\\textbf{ (1.6$\times$)}}}&{\makecell[c]{\textbf{128.4}\\\textbf{(1.5$\times$)}}}\\
\Xhline{1pt}
\end{tabular}}}
%\vspace{-5pt}
\end{table}
TABLE~\ref{Complexity of Matrix-Vector Multiplication} describes the comparison of the overhead of \Name and  Overdrive in generating matrix-vector multiplication triples in different dimensions. It is clear that \Name is  superior in performance to Overdrive, both in terms of communication overhead and computational overhead. We observe that \Name achieves more than $8\times$ reduction in communication overhead and at least $1.5\times$ speedup in computation compared to Overdrive. This stems from Overdrive's disadvantage in constructing triples, i.e. constructing triples for only a single multiplication operation (or multiplication between a single row of a matrix and a vector). In addition, the generation process requires frequent interaction between the client and the model holder (for zero-knowledge proofs and  preventing breaches by either party). This inevitably incurs substantial computational and communication overhead. Our constructed matrix-multiplication triples enable the communication overhead to be independent of the number of  multiplications, only related to the size of the input. This substantially reduces the amount of data that needs to be exchanged between $P_0$ and $P_1$. In addition, we move the majority of the computation to be executed by  $P_1$, which avoids the need for distributed decryption and frequent zero-knowledge proofs in malicious adversary settings. Moreover,  our matrix-vector multiplication does not involve any rotation operation. As a result, these optimization methods motivate \Name to exhibit a satisfactory performance overhead in generating triples.

\subsubsection{\quad \; Cost of generating convolution triple}
\begin{table}[htb]
\centering
\small
%\vspace{-10pt}
\caption{Cost of generating the convolution triple}
\label{Complexity of convolution}
\setlength{\tabcolsep}{1.6mm}{
\resizebox{\linewidth}{!}{\begin{tabular}{|c|c|c|c|c|c|c|c|}
\Xhline{1pt}
\multirow{3}*{\textbf{Input}}&\multirow{3}*{\textbf{Kernel}}&\multicolumn{2}{c|}{\textbf{Comm.cost (GB)}} &\multicolumn{4}{c|}{\textbf{Running time (s)} }\\ \cline{3-8}
&&\multirow{2}*{Overdrive}&\multirow{2}*{\textbf{\Name}} & \multicolumn{2}{c|}{Overdrive} & \multicolumn{2}{c|}{ \Name (Speedup)} \\
& & & & LAN & WAN & LAN & WAN\\
\Xhline{1pt}
{\makecell[c]{$16\times16$ \\$@128$}}& \makecell[c]{$1\times1$ \\$@128$}& 17.1& \textbf{2.1}&{1476.1}&{1494.6}&{\makecell[c]{\textbf{924.7}\\\textbf{ (1.6$\times$)}}}&{\makecell[c]{\textbf{938.4}\\\textbf{(1.6$\times$)}}}\\ \hline
{\makecell[c]{$16\times16$ \\$@256$}}& \makecell[c]{$1\times1$ \\$@256$}& 67.8& \textbf{8.2}&{6059.3}&{6059.31}&{\makecell[c]{\textbf{3568.8}\\\textbf{ (1.7$\times$)}}}&{\makecell[c]{\textbf{3580.8}\\\textbf{(1.7$\times$)}}}\\ \hline
{\makecell[c]{$16\times16$ \\$@512$}}& \makecell[c]{$3\times3$ \\$@128$}& 467.5& \textbf{56.8}&{40753.4}&{40767.1}&{\makecell[c]{\textbf{25387.2}\\\textbf{ (1.6$\times$)}}}&{\makecell[c]{\textbf{25401.5}\\\textbf{(1.6$\times$)}}}\\ \hline
{\makecell[c]{$32\times32$ \\$@2048$}}& \makecell[c]{$5\times5$ \\$@512$}& 83127.8& \textbf{7324.3}&{7245056.2}&{7245068.8}&{\makecell[c]{\textbf{4521023.3}\\\textbf{ (1.6$\times$)}}}&{\makecell[c]{\textbf{4521165.6}\\\textbf{(1.6$\times$)}}}\\ \Xhline{1pt}
\end{tabular}}}
%\vspace{-5pt}
\end{table}
TABLE~\ref{Complexity of convolution} shows the comparison of the performance of \Name and Overdrive in generating convolution triples in different dimensions, where input tensor of size $u_w\times u_h$ with $c_i$ channels  is denoted as $u_w\times u_h @c_i$,  and the size of corresponding  kernel is denoted as  $k_w\times k_h @c_o$. We observe that \Name is much lower than Overdrive in terms of computational and communication overhead. For instance, \Name gains a reduction of up to  $9\times$ in communication cost and a speedup of at least $1.6\times$ in computation. This is due to the optimization method customized by \Name for generating convolution triples. Compared to Overdrive, which focuses on constructing authenticated triples for a single multiplication operation, \Name uses the homomorphic parallel matrix multiplication method constructed in \cite{jiang2018secure} as the underlying structure to construct matrix multiplication triples equivalent to convolution triples.  Since a single matrix is regarded as a computational entity, the above method makes the communication overhead between the client and the model holder only related to the size of the matrix, and independent of the number of operations of the multiplication between the two matrices (that is, the communication complexity is reduced from $O(d^3)$ to $O(d)^2$  given the multiplication between the two $d\times d$ matrices). In addition, the optimized parallel matrix multiplication reduces the homomorphic rotation operation from $O(d^2)$ to $O(d)$. This enables \Name to show significant superiority in computing convolution triples.
% \subsubsection{Cost for preprocessing the nonlinear layer}
% \setcounter{figure}{7}
% \begin{figure}[htb]
%   \centering
%   \includegraphics[width=0.35\textwidth]{figure/Pcommunication.pdf}
%  % \vspace{-10pt}
%   \caption{Comparison of communication overhead for preprocessing the nonlinear layer.}
%   \label{PNon-layer cost}
%   %\vspace{-10pt}
% \end{figure}
% In the preprocessing stage, the arithmetic operations performed by \Name and Overdrive are almost identical, namely constructing garbled circuits and sending them to the model holder. Therefore, here we omit the comparison of the computational cost of the two  and focus on analyzing their difference in communication cost. Figure~\ref{PNon-layer cost} provides a comparison of the preprocessing cost of \Name and Overdrive on different numbers of ReLUs. We observe
% a $14\times$ reduction in communication overhead for\Name compared to Overdrive, which stems from the optimized $\mathtt{GC}$ constructed by \Name for each ReLU function. In more detail, given the $\lambda=128$ bits security parameter, 44-bit prime field, and arbitrary  input $\mathbf{v}_i$, \Name only construct the circuit for the nonlinear part of the ReLU, and the output of the GC is each bit  of  $sign(\mathbf{v}_i)$'s  share (i.e., $sign(\mathbf{v}_i)[j]$, for $1\leq j\leq \kappa$), rather than obtaining the exact arithmetic share. Hence, compared with Overdrive, \Name reduces the communication overhead of each ReLU function from $2c\lambda+190\kappa\lambda+232\kappa^2$ to $2d\lambda+4\kappa\lambda+6\kappa^2$, where $d\ll c$.

\subsection{Performance of online phase}
\label{Performance of online phase}
In the online phase, \Name is required to perform operations at the linear and nonlinear layers alternately. Here we discuss the overhead performance of \Name compared to Overdrive separately.
\subsubsection{\quad \; Performance of executing linear layers}

\begin{table}[htb]
\centering
\footnotesize
%\vspace{-10pt}
\caption{Comparison of the communication overhead  for executing  convolution in the online phase}
\label{online-linear layer costss}
{\begin{tabular}{|c|c|c|c|}
\Xhline{1pt}
\multirow{2}*{\textbf{Input}}&\multirow{2}*{\textbf{Kernel}}&\multicolumn{2}{c|}{\textbf{Comm.cost (MB)}} \\ \cline{3-4}
&&{Overdrive}&{\Name (Reduction)} \\
\Xhline{1pt}
{\makecell[c]{$16\times16$ \\$@128$}}& \makecell[c]{$1\times1$ \\$@128$}& 46.1&{\makecell[c]{\textbf{0.5}\\\textbf{ (85.3$\times$)}}}\\ \hline
{\makecell[c]{$16\times16$ \\$@256$}}& \makecell[c]{$1\times1$ \\$@256$}& 184.5&{\makecell[c]{\textbf{1.4}\\\textbf{(128.0$\times$)}}}\\ \hline
{\makecell[c]{$16\times16$ \\$@512$}}& \makecell[c]{$3\times3$ \\$@128$}& 1271.7&{\makecell[c]{\textbf{15.7}\\\textbf{(81.2$\times$)}}}\\ \hline
{\makecell[c]{$32\times32$ \\$@2048$}}& \makecell[c]{$5\times5$ \\$@512$}& 226073.0&{\makecell[c]{\textbf{1459.8}\\\textbf{(154.9$\times$)}}}\\ \Xhline{1pt}
\end{tabular}}
%\vspace{-5pt}
\end{table}
Since both \Name and Overdrive follow the same computational logic to perform the linear layer in the online phase, i.e. use pre-generated authenticated triples to compute matrix-vector multiplication and convolution, both exhibit similar compu-
tational overhe. Therefore, we focus on analyzing the difference in communication overhead between the two of executing convolution. TABLE~\ref{online-linear layer costss} depicts the communication overhead of \Name and Overdriveffor computing convolution in different dimensions. It is obvious that \Name shows superior performance in communication overhead compared to Overdrive. This is mainly due to the fact that Overdrive needs to open a fresh authenticated Beaver's multiplication triple for each multiplication operation, which makes the communication overhead of executing the entire linear layer positively related to the total multiplication operations involved. In contrast, \Name customizes matrix-vector multiplication and convolution triples, which makes the cost independent of the number of multiplication operations in the linear layer. This substantially reduces the amount of data that needs to be exchanged during the execution.

\subsubsection{\quad \; Performance of executing nonlinear layers}
\begin{figure}[htb]
  \centering
 % \vspace{-10pt}
  \subfigure[]{\label{FIGa}\includegraphics[width=0.23\textwidth]{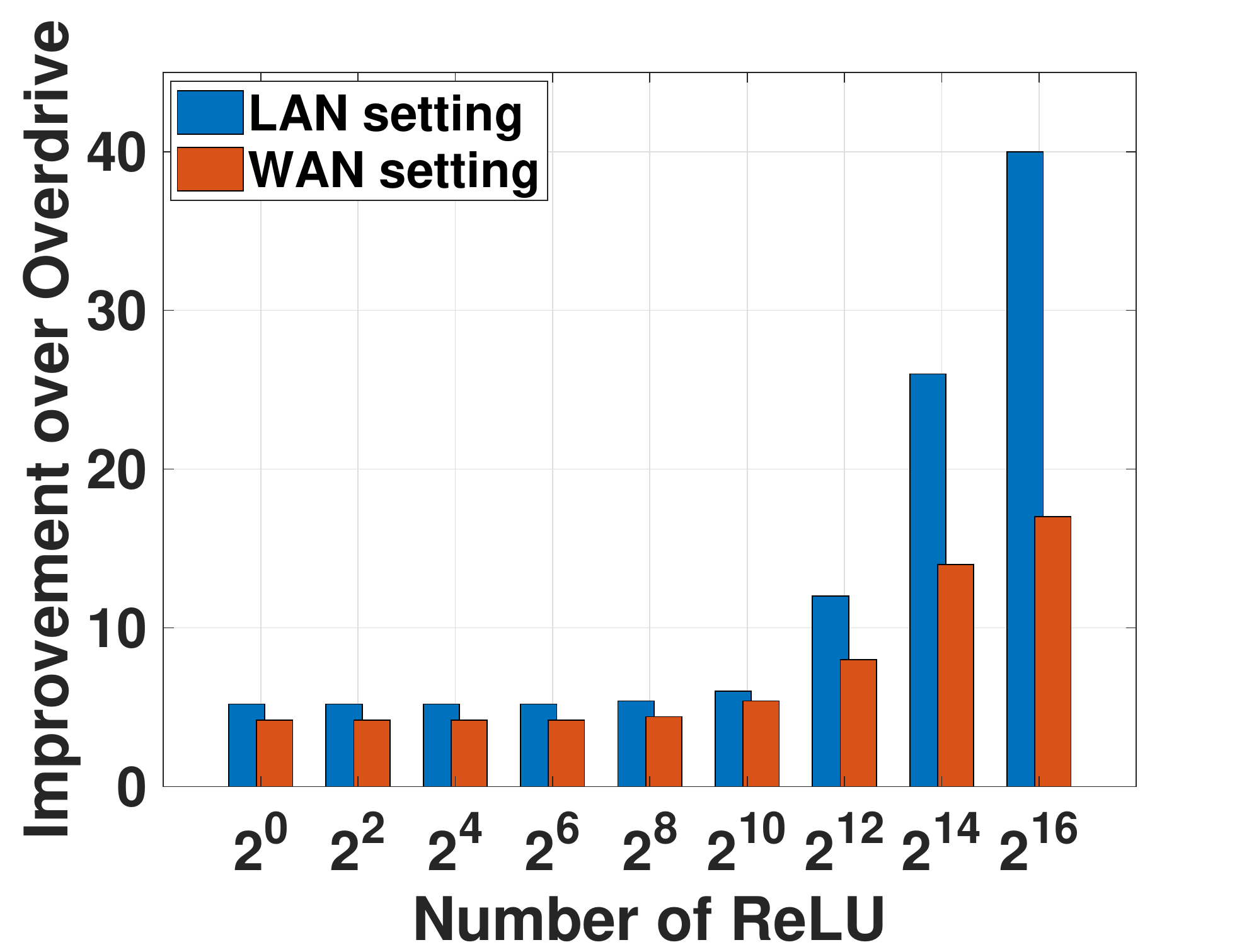}}
 \subfigure[]{\label{FIGb}\includegraphics[width=0.23\textwidth]{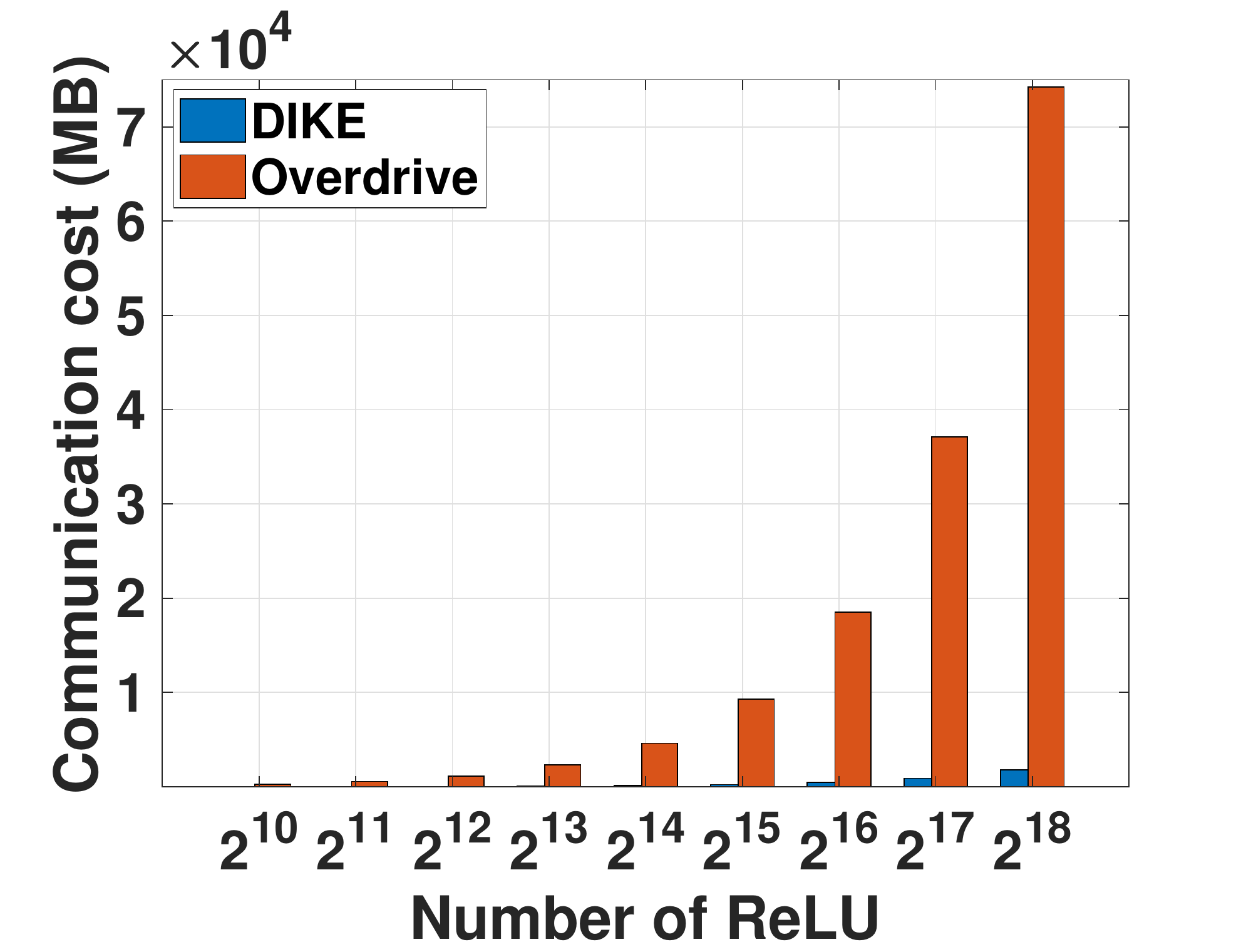}}
  \caption{Comparison of the overhead  for executing nonlinear layers. ((a) Running time improvement of \Name over Overdrive. The y-axis shows $ \rm \frac{Overdrive\; time}{\Name \; time}$  (b) Comparison of the communication overhead.}
  \label{online-linear layer cost}
  %\vspace{-10pt}
\end{figure}
%Since \Name  and Overdrive follow a similar method of exchanging data during the online phase of executing the nonlinear layer, that is, to obtain the input labels of the $\mathtt{GC}$ through the malicious OT protocol, they also exhibit similar communication overhead. For  brevity, we focus on discussing the computation overhead incurred by executing different numbers of ReLUs.
Figure~\ref{online-linear layer cost} provides the comparison of the cost between Overdrive and \Name. We observe that \Name outperforms Overdrive by $4-42\times$ in runtime on LAN Setting and $3-16\times$ in WAN Setting. For example, Overdrive takes $165.4s$ and $1283.5s$ to compute $2^{15}$ ReLUs on LAN  and WAN setting, respectively. Whereas, \Name took just $5.1s$ and $110.2s$ in the respective network settings. For communication overhead, we observed that Overdrive required $401$KB of traffic to perform a single ReLU while we only need $8.33$KB, which is at least a $48\times$ improvement.  This is mainly due to the fact that our optimized $\mathtt{GC}$ substantially reduces the multiplication operations involved in evaluating in the $\mathtt{GC}$. Moreover, Overdrive needs to verify the correctness of the input from the model holder in the $\mathtt{GC}$, which is very expensive. Conversely, \Name designs lightweight consistency verification methods to achieve this.
\subsection{Performance of end-to-end secure inference}
\label{Secure Inference Performance}
\begin{table}[htb]
\centering
\small
%\vspace{-5pt}
\caption{Cost of  end-to-end secure inference}
\label{Complexity of  end-to-end secure inference}
\setlength{\tabcolsep}{1.9mm}{
\resizebox{\linewidth}{!}{\begin{tabular}{|c|c|c|c|c|c|c|c|}
\Xhline{1pt}
\multirow{9}*{\textbf{LeNet}}&\multirow{3}*{\textbf{Phases}}&\multicolumn{2}{c|}{\textbf{Comm.cost (MB)}} &\multicolumn{4}{c|}{\textbf{Running time (s)} }\\ \cline{3-8}
&&\multirow{2}*{Overdrive}&\multirow{2}*{\textbf{\Name}} & \multicolumn{2}{c|}{Overdrive} & \multicolumn{2}{c|}{ \Name (Speedup)} \\
&& & &  LAN & WAN & LAN & WAN\\  \cline{2-8}
&{Offline }& 3427.8& \textbf{209.6}&{235.5}&{246.8}&{\makecell[c]{\textbf{92.9}\\\textbf{ (2.5$\times$)}}}&{\makecell[c]{\textbf{104.6}\\\textbf{(2.4$\times$)}}}\\ \cline{2-8}
&{Online}& 2543.1& \textbf{54.0}&{32.8}&{254.9}&{\makecell[c]{\textbf{1.0}\\\textbf{ (32.6$\times$)}}}&{\makecell[c]{\textbf{21.9}\\\textbf{(11.6$\times$)}}}\\\cline{2-8}
&{Total}& 5970.9& \textbf{263.6}&{268.3}&{501.7}&{\makecell[c]{\textbf{93.9}\\\textbf{ (2.9$\times$)}}}&{\makecell[c]{\textbf{126.5}\\\textbf{(4.0$\times$)}}}\\ \Xhline{1pt}
\multirow{9}*{\textbf{ResNet18}}&\multirow{3}*{\textbf{Phases}}&\multicolumn{2}{c|}{\textbf{Comm.cost (MB)}} &\multicolumn{4}{c|}{\textbf{Running time (s)} }\\ \cline{3-8}
&&\multirow{2}*{Overdrive}&\multirow{2}*{\textbf{\Name}} & \multicolumn{2}{c|}{Overdrive} & \multicolumn{2}{c|}{ \Name (Speedup)} \\
&& & &  LAN & WAN & LAN & WAN\\  \cline{2-8}
&{Offline }& 2116018.6& \textbf{257257.7}&{238774.2}&{238957.4}&{\makecell[c]{\textbf{114003.1}\\\textbf{ (2.1$\times$)}}}&{\makecell[c]{\textbf{114978.8}\\\textbf{(2.1$\times$)}}}\\ \cline{2-8}
&{Online}& 19359.5& \textbf{459.4}&{177.0}&{1373.7}&{\makecell[c]{\textbf{5.5}\\\textbf{ (32.2$\times$)}}}&{\makecell[c]{\textbf{117.9}\\\textbf{(11.7$\times$)}}}\\\cline{2-8}
&{Total}& 2135378.1& \textbf{257717.1}&{238951.2}&{240331.1}&{\makecell[c]{\textbf{114008.6}\\\textbf{ (2.1$\times$)}}}&{\makecell[c]{\textbf{115096.7}\\\textbf{(2.1$\times$)}}}\\ \Xhline{1pt}
\end{tabular}}}
%\vspace{-5pt}
\end{table}
We compare the performance of \Name and Overdrive on real-world ML models. In our experiments, we choose ResNet-18 and LeNet, which are trained on the CelebA \cite{liu2015deep} and  C-MNIST datasets\cite{arjovsky2019invariant} respectively. Note that CelebA and  C-MNIST are widely used to check how fair a given trained model is. TABLE~\ref{Complexity of  end-to-end secure inference} shows the performance of \Name and Overdrive in terms of computation and communication overhead. Compared to Overdrive,  \Name  demonstrates an encouraging online runtime boost by $32.6\times$ and  $32.2\times$ over existing works  on LeNet and ResNet-18, respectively,  and at least  an order of magnitude  communication cost reduction. In online phase,  Overdrive takes $32.8s$ and $177s$ to compute single query on LeNet and ReNet-18, respectively. Whereas, \Name took just $1s$ and $5.5s$ in the respective network settings.  Consistent with the previous analysis, this stems from the customized optimization mechanism we designed for \Name.

\subsection{Comparison with other works}
\noindent\textbf{Compared with DELPHI}. We demonstrate that for the execution of non-linear layers, the communication overhead of \Name is even lower than the state-of-the-art scheme DELPHI \cite{mishra2020delphi} under hte semi-honest threat model. Specifically, for the $i$-th nonlinear layer,  DELPHI needs to calculate shares of  $f_i(\mathbf{v_i})$ in $\mathtt{GC}$ and share it with two parties. DELPHI requires at least $3\kappa$ additional AND gates, which incurs at least $6\kappa\lambda$ bits of communication, compared to only computing each bit of $f_i(\mathbf{v_i})$ in \Name. In our experiment, For $\kappa=44$, $\lambda=28$, our method gives roughly $9\times$ less communication for generating shares of $f_i(\mathbf{v_i})$, i.e.,  DELPHI required $32$KB of traffic to perform a single ReLU while we only need $8.33$KB.

\noindent\textbf{Compared with MUSE and SIMC}. We note that several works such as MUSE \cite{lehmkuhl2021muse} and SIMC \cite{chandran2021simc} have been proposed to address ML secure inference on the \textit{client malicious} threat model. Such a threat model considers that the server (i.e., the model holder) is semi-honest but the malicious client may arbitrarily violate the protocol to obtain private information. These works intuitively seem to translate to our application scenarios with appropriate modification. However, we argue that this is non-trivial. In more detail, in the \textit{client malicious} model, the client's inputs are encrypted  and sent to the semi-honest model holder, which performs all linear operations for speeding up  the computation. Since the model holder holds the model parameter in the plaintext, executing the linear layer only involves homomorphic operations between the plaintext and the ciphertext. Such type of computation is compatible with mainstream homomorphic optimization methods including GALA \cite{zhang2021gala} and GAZELLE \cite{juvekar2018gazelle}. However, in \Name, the linear layer operation cannot be done in the model holder because it is considered malicious. One possible approach is to encrypt the model data and perform linear layer operations with two-party interaction. This is essentially performing  homomorphic operations between ciphertext and ciphertext, which is not compatible with previous optimization strategies. Therefore, instead of simply fine-tuning MUSE \cite{lehmkuhl2021muse} and SIMC \cite{chandran2021simc}, we must redesign new parallel homomorphic computation methods to fit this new threat model. On the other hand, we observe that the techniques for nonlinear operations in MUSE \cite{lehmkuhl2021muse} and SIMC \cite{chandran2021simc} can clearly be transferred to \Name. However, our method still outperforms SIMC (an upgraded version of MUSE). This mainly stems from the fact that we only encapsulate the nonlinear part of ReLU into $\mathtt{GC}$ to further reduce the number of multiplication operations. Experiments show that our method is about one third of SIMC in terms of computing and communication overhead.

\section{Conclusion}
\label{sec:conclusion}
In this paper, we proposed \Name, the first secure inference framework  to check the fairness degree of a given ML model. We designed a series of optimization methods to reduce the
overhead of the offline stage. We also presented optimized $\mathtt{GC}$ to substantially speed up operations in the non-linear layers. In the future, we will focus on designing more efficient optimization strategies to further reduce the computation overhead of \Name, to make secure ML inference more suitable for a wider range of practical applications.
%\clearpage
\bibliographystyle{plain}
\balance
\bibliography{PPDR}
\clearpage
\appendices
\section*{Appendix}
\setcounter{section}{0}

\section{Threat Model}
\label{A:threat model}
We formalize the threat model involved in \Name with the simulated paradigm \cite{lindell2017simulate}. We define two interactions to capture security: a real interaction by $P_0$ and $P_1$ in the presence of adversary $\mathbf{A}$ and an environment $Z$, and an ideal interaction where parties send their respective inputs to a trusted entity that computes functionally faithfully. Security requires that for any adversary $\mathbf{A}$ in real interaction, there exists a simulator $\mathbf{S}$ in ideal interaction, such that no environment $\mathbf{Z}$ can distinguish real interaction from ideal interaction. Specifically, let $f=(f_0, f_1)$  be the two-party functionality such that $P_0$ and $P_1$ invoke $f$ on inputs $a$ and $b$ to obtain $f_0(a, b)$ and $f_1(a, b)$, respectively.  We say a protocol $\pi$  securely implements $f$ if it holds the following properties.
\begin{itemize}

\item  \textbf{Correctness}: If $P_0$ and $P_1$ are both honest, then $P_0$ gets $f_0(a, b)$ and $P_1$ gets $f_1(a, b)$ from the execution of $\pi$ on the inputs  $a$ and $b$, respectively.
\item \textbf{Semi-honest Client Security}: For a semi-honest adversary $\mathbf{A}$ that compromises $P_1$, there exists a simulator $\mathbf{S}$ such that for any input $(a, b)$,   we have
\begin{small}
\begin{equation*}
\begin{split}
View_{\mathbf{A}}^{\pi}(a, b)\approx \mathbf{S}(b, f_1(a, b))
\end{split}
\end{equation*}
\end{small}
where $View_{\mathbf{A}}^{\pi}(a, b)$  represents the view of $\mathbf{A}$  during the execution of $\pi$, and $a$ and $b$ are the inputs of $P_0$ and $P_1$, respectively.
  $\mathbf{S}(b, f_1(a, b))$  represents the view simulated by $\mathbf{S}$ when it is given access to  $b$ and  $f_1(a, b)$. $\approx$ indicates computational indistinguishability
of   two distributions $View_{\mathbf{A}}^{\pi}(a, b)$ and $\mathbf{S}(b, f_1(a, b))$.
 \item \textbf{Malicious Model Holder Security}: For the malicious  adversary $\mathbf{A}$ that compromises $P_0$, there exists a simulator $\mathbf{S}$, such that for any input $b$ from $P_1$,  we have
 \begin{small}
\begin{equation*}
\begin{split}
Out_{P_1}, View_{\mathbf{A}}^{\pi}(b,\cdot)\approx \hat{Out}, \mathbf{S}^{f(b, \cdot)}
\end{split}
\end{equation*}
\end{small}
\end{itemize}
where $View_{\mathbf{A}}^{\pi}(b, \cdot)$ denotes $\mathbf{A}$'s view during the execution of $\pi$ with  $S_1$'s input $b$. $Out_{P_1}$ indicates the output of $P_1$ in the real protocol execution. Similarly,   $\hat{Out}$  and $\mathbf{S}^{f(b, \cdot)}$  represents the output of $P_1$  and the simulated view  in the ideal interaction.

\section{Proof of Theorem 1}
\label{Proof of Theorem 1}
\renewcommand\tablename{Figure}
\renewcommand \thetable{\arabic{table}}
\setcounter{table}{7}
\begin{proof}
\begin{table}
\begin{tabular}{|p{8cm}|}
\hline
 \textbf{Input:} $P_0$ holds $\langle \mathbf{X}\rangle_0$  uniformly chosen from $\mathbb{F}_p^{d_1\times d_2}$ and $ \langle \mathbf{y}\rangle_0 $ uniformly chosen from $\mathbb{F}_p^{d_2}$.   $P_1$ holds $\langle \mathbf{X}\rangle_1$  uniformly chosen from $\mathbb{F}_p^{d_1\times d_2}$, and $ \langle \mathbf{y}\rangle_1$ uniformly chosen from $\mathbb{F}_p^{d_2}$ and a MAC key $\alpha$ uniformly chosen from $\mathbb{F}_p$\\
 \textbf{Output:} $P_b$ obtains $\{[\![\mathbf{X}]\!]_b, [\![\mathbf{y}]\!]_b, [\![\mathbf{z}]\!]_b\}_{b\in \{0, 1\}}$, where $\mathbf{X}\ast \mathbf{y}=\mathbf{z}$.\\
\hline
\end{tabular}
\caption{Functionality of $\mathcal{F}_{Mtriple}$}
\label{Functionality FMriple}
\end{table}
Let $\mathcal{F}_{Mtriple}$ shown in Figure~\ref{Functionality FMriple} be the functionality of generating matrix-vector multiplication triple. We first prove security for semi-honest clients and then demonstrate security against malicious model holders.

\noindent\textbf{Semi-honest client security}. The simulator $\mathtt{Sim_c}$ samples $(pk, sk)\leftarrow \mathtt{KeyGen}(1^\lambda)$. The simulator and the semi-honest client run a secure two-party protocol to generate the public and secret keys for homomorphic encryption. When the  simulator accesses the ideal functionality, it provides $pk$ as output. In addition, the $\mathtt{Sim_c}$ sends $\mathtt{Enc}_{pk}(\mathbf{0})$ to the client along with the simulated zero-knowledge proof of well-formedness of ciphertexts. We now show the indistinguishability between real and simulated views by the following hybrid arguments.
\begin{itemize}
\item  $\mathtt{Hyb_1}$: This corresponds to the real execution of the protocol.
\item $\mathtt{Hyb_2}$: The simulator $\mathtt{Sim_c}$  runs the two-party computation protocol with the semi-honest client to generate the public and secret keys for homomorphic encryption. When the  simulator accesses the ideal functionality, we sample  $(pk, sk)\leftarrow \mathtt{KeyGen}(1^\lambda)$ and send $pk$ to the  semi-honest client. This hybrid is computationally indistinguishable to $\mathtt{hyb_1}$.
 \item $\mathtt{Hyb_3}$:  In this hybrid, instead of  sending the encryptions $c_1\leftarrow \mathtt{Enc}(pk, \langle \mathbf{X}\rangle_0)$  and $c_2\leftarrow \mathtt{Enc}(pk, \langle \mathbf{Y}\rangle_0)$ to $P_1$, $\mathtt{Sim_c}$ sends ciphertexts with all 0s (i.e., $\mathtt{Enc}_{pk}(\mathbf{0})$) to the client. $\mathtt{Sim_c}$ also provides  a zero-knowledge (ZK) proof of plaintext knowledge of the ciphertexts. For any two plaintexts, FHE ensures that an adversary cannot distinguish them from their ciphertexts. In addition, zero-knowledge proofs also guarantee the indistinguishability of two ciphertexts. Therefore, this hybrid is indistinguishable from the previous one.
\end{itemize}
\noindent\textbf{Malicious model holder security}. The simulator $\mathtt{Sim_m}$ samples $(pk, sk)\leftarrow \mathtt{KeyGen}(1^\lambda)$. The simulator and the semi-honest client run a secure two-party protocol to generate the public and secret keys for homomorphic encryption. When the  simulator accesses the ideal functionality, it provides $(pk, sk)$ as outputs. Once $P_0$ sends $c_1\leftarrow \mathtt{Enc}(pk, \langle \mathbf{X}\rangle_0)$  and $c_2\leftarrow \mathtt{Enc}(pk, \langle \mathbf{Y}\rangle_0)$,  $\mathtt{Sim_m}$ verifies the validity of the ciphertext from the client. If the verification is passed, $\mathtt{Sim_m}$   extracts $\langle \mathbf{X}\rangle_0$  and $\langle \mathbf{Y}\rangle_0$  and the randomness used for generating
these ciphertexts,  since it  has access to the client's input. Then, $\mathtt{Sim_m}$  samples $\langle \mathbf{X}\rangle_1$ and $\langle \mathbf{Y}\rangle_1$, and queries the ideal functionalities on the input $\langle \mathbf{X}\rangle_0$, $\langle \mathbf{X}\rangle_1$, $\langle \mathbf{Y}\rangle_0$ and  $\langle \mathbf{Y}\rangle_1$ to obtain  $(\langle \alpha \mathbf{X}\rangle_0, \langle \alpha \mathbf{y}\rangle_0, \langle \alpha \mathbf{z}\rangle_0, \langle \mathbf{z} \rangle_0)$. Then, $\mathtt{Sim_m}$ uses these outputs and the randomness used to generate the initial ciphertexts to construct  the four simulated ciphertexts. It sends the simulated ciphertexts to the client.
\begin{itemize}
\item  $\mathtt{Hyb_1}$: This corresponds to the real execution of the protocol.
\item $\mathtt{Hyb_2}$: The simulator $\mathtt{Sim_m}$  runs the two-party computation protocol with the malicious model holder to generate the public and secret keys for homomorphic encryption. When the  simulator accesses the ideal functionality, we sample  $(pk, sk)\leftarrow \mathtt{KeyGen}(1^\lambda)$ and send them to the  malicious model holder. This hybrid is computationally indistinguishable to $\mathtt{hyb_1}$.
 \item $\mathtt{Hyb_3}$:  In this hybrid, $\mathtt{Sim_m}$ checks the validity of the ciphertext from the client. If the zero-knowledge proofs are valid, $\mathtt{Sim_m}$   extracts $\langle \mathbf{X}\rangle_0$  and $\langle \mathbf{Y}\rangle_0$  and the randomness used for generating these ciphertexts,  since it  has access to the client's input. The properties of zero-knowledge proofs ensure that this hybrid  is indistinguishable from the previous one.
 \item $\mathtt{Hyb_4}$:   $\mathtt{Sim_m}$  exploits the functional privacy of FHE to generate  $c_3=\mathtt{Enc}_{pk}( \alpha(\langle \mathbf{X}\rangle_1+\langle \mathbf{X}\rangle_0)-\langle \alpha \mathbf{X}\rangle_1)$, $c_4=\mathtt{Enc}_{pk}( \alpha(\langle \mathbf{Y}\rangle_1+\langle \mathbf{Y}\rangle_0)-\langle \alpha \mathbf{Y}\rangle_1)$, $c_5=\mathtt{Enc}_{pk}(\alpha(\mathbf{X}\odot \mathbf{Y})-\langle \alpha \mathbf{Z}\rangle_1)$, and $c_6=\mathtt{Enc}_{pk}((\mathbf{X}\odot \mathbf{Y})-\langle \mathbf{Z}\rangle_1)$. This hybrid is computationally indistinguishable
to the previous hybrid from the function privacy of the FHE scheme. Note that view of the model holder in $\mathtt{Hyb_4}$ is identical to the view generated by $\mathtt{Sim_m}$.

\end{itemize}
\end{proof}

\section{ Conversion between convolution and matrix multiplication}
\label{ Conversion between convolution and matrix multiplication}
\setcounter{figure}{8}
\begin{figure}[htb]
\centering
\includegraphics[width=0.5\textwidth]{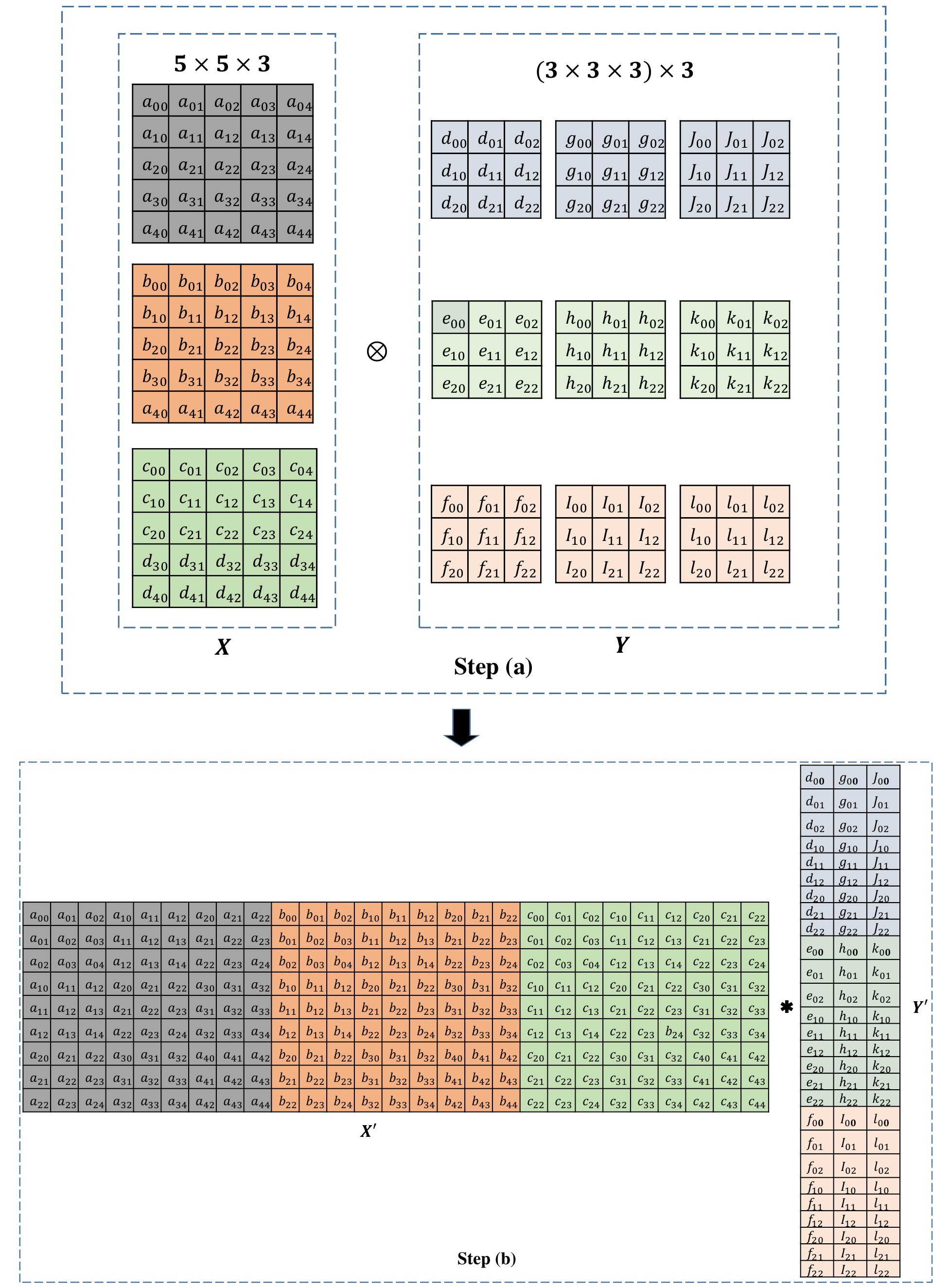}
\caption{Conversion between convolution and matrix multiplication}
\label{Conversion between convolution and matrix multiplication}
\end{figure}
\setcounter{table}{8}
Figure~\ref{Conversion between convolution and matrix multiplication} provides an example to convert a given  convolution into the corresponding matrix multiplication.  As shown in Figure~\ref{Conversion between convolution and matrix multiplication}, given  input tensor of size $ 5\times 5$ with $3$ channels, denoted as $\mathbf{X}$, $3$  kernels with  a size of $(2+1)\times (2+1)\times 3$ denote as tensor $\mathbf{Y}$,   the convolution between $\mathbf{X}$ and $\mathbf{Y}$ are  converted  an equivalent matrix multiplication $\mathbf{X'}$ and $\mathbf{Y'}$, where the number of turns to zero-pad is $0$, and stride $s=1$. Specifically,  we construct a matrix $\mathbf{X'}$ with dimension $(9\times 27)$, where $\mathbf{X'}_{(i, j)(\Delta_i, \Delta_j,k)}=\mathbf{X}_{i+\Delta_i, j+\Delta_j, k}$. Similarly, we construct a matrix $\mathbf{Y'}$ of dimension $(27 \times 3)$ such that $\mathbf{Y'}(\Delta_i, \Delta_j, k)k'=\mathbf{Y}_{\Delta_i, \Delta_j, k'}$. Then, the original convolution operation is transformed into $\mathbf{Z'}= \mathbf{X'}\ast \mathbf{Y'}$, where $\mathbf{Z'}_{(ij)k'}$ $=\mathbf{Z}_{ijk'}$.

\section{Proof of Theorem 4}
\label{Proof of Theorem 4}
\begin{proof}
\noindent\textbf{Semi-honest client security}.The security of the protocol $\pi_{ONlin}$ against the semi-honest client $P_1$ is evident by observing the execution of the protocol. This stems from the fact that $P_1$ does  obtain output in OT$_{\lambda}^{\kappa}$ and does not receive any information from $P_0$ in subsequent executions. Here we focus on the security analysis of $\pi_{ONlin}$ against malicious model holder $P_0$.

\noindent\textbf{Malicious model holder security}.  We first define the functionality of the protocol  $\pi_{ONlin}$, denoted as $\mathcal{F}_{ONlin}$, as shown in Figure~\ref{Functionality of the nonlinear layer}. We use $\mathtt{Real}$ to refer to the view of the real interaction between $P_1$ and the adversary $\mathcal{A}$ controlling $P_0$, and then demonstrate $\mathtt{Real}$  indistinguishability from the simulated view interacted by the simulator $\mathtt{Sim_m}$ and $\mathcal{A}$ through standard hybrid arguments. In the following we will define three hybrid executions $\mathtt{Hyb_1}$, $\mathtt{Hyb_2}$ and $\mathtt{Hyb_3}$. We  prove that $\pi_{ONlin}$ is secure from the malicious model holder $P_0$ by proving indistinguishability among these hybrid executions.
\setcounter{table}{9}
\begin{table}[htb]
\centering
\footnotesize
%\small
\begin{tabular}{|p{8cm}|}
\hline
Function $f: \mathbb{F}_p \rightarrow \mathbb{F}_p$.\\
 \textbf{Input:} $P_1$ holds $\left \langle \mathbf{v}_i \right \rangle_1 \in \mathbb{F}_{p}$  and  a MAC key $\alpha$ uniformly chosen from $\mathbb{F}_p$.  $P_0$ holds  $\left \langle \mathbf{v}_i \right \rangle_0 \in \mathbb{F}_{p}$.\\
 \textbf{Output:} $P_b$ obtains $\{(\langle \alpha\mathbf{v}_i\rangle_b, \langle f(\mathbf{v}_i)\rangle_b, \langle \alpha f(\mathbf{v}_i)\rangle_b)\}$ for $b\in\{0, 1\}$.\\
\hline
\end{tabular}
\caption{Functionality of the nonlinear layer $\mathcal{F}_{ONlin}$}
\label{Functionality of the nonlinear layer}
\end{table}

$\mathtt{Hyb_1}$: This hybrid execution is identical to $\mathtt{Real}$ except in the authentication phase.  To be precise, in the authentication phase, the simulator $\mathtt{Sim_m}$  use labels $\hat{\mathtt{lab}_{i, j}^{out}}$ (described below) to replace the labels $\mathtt{lab}_{i, j}^{out}$ used in $\mathtt{Real}$. Please note that in this hybrid the simulator $\mathtt{Sim_m}$  can access $P_1$' input $\left \langle \mathbf{v}_i \right \rangle_1$ and $\alpha$, where $\left \langle \mathbf{v}_i \right \rangle_0 +\left \langle \mathbf{v}_i \right \rangle_1= \mathbf{v}_i$. Let $\delta=(\mathbf{v}_i||sign(\mathbf{v}_i))$. Therefore, for $i\in[2\kappa]$, we set  $\hat{\mathtt{lab}_{i, j}^{out}}=\mathtt{lab}_{i, j}^{out}$ if $j=\delta[i]$, otherwise, $\hat{\mathtt{lab}_{i, 1-\delta[i]}^{out}}$ (\textit{i.e.}, the ``other" label) is set to a random value chosen from $\{0, 1\}^{\lambda}$ uniformly, where the first bit of  $\hat{\mathtt{lab}_{i, 1-\delta[i]}^{out}}$ is $1-\varsigma_{i, \delta[i]}$. We provide the formal description of $\mathtt{Hyb_1}$ as follows, where the  indistinguishability between the view of $\mathcal{A}$ in $\mathtt{Real}$ and $\mathtt{Hyb_1}$ is directly derived from the authenticity of the garbled circuit.
\begin{itemize}
    \item[1.] $\mathtt{Sim_m}$ receives $\left \langle \mathbf{v}_i \right \rangle_0$ from $\mathcal{A}$ as the input of OT$_{\lambda}^{\kappa}$.
\item[2.] Garbled Circuit Phase:
\begin{itemize}
\item For   $booln^f$, $\mathtt{Sim_m}$ first computes $\mathtt{Garble}(1^\lambda, booln^f)\rightarrow (\mathtt{GC}, \{ \{\mathtt{lab}_{i,j}^{in}\},\{\mathtt{lab}_{i,j}^{out}\}\}_{j\in\{0,1\}})$ for each $i\in[2\kappa]$,  and then for ${i\in\{\kappa+1, \cdots, 2\kappa\}}$ sends  $\{\mathtt{\tilde{lab}}_{j}^{in}=\mathtt{lab}_{j, \left \langle \mathbf{v}_i \right \rangle_0[j]}^{in}\}$ to $\mathcal{A}$ as the output of  OT$_{\lambda}^{\kappa}$. In addition, $\mathtt{Sim_m}$ sends the garbled  circuit $\mathtt{GC}$ and its garbled inputs $\{ \{\mathtt{\tilde{lab}}_{j}^{in}=\mathtt{lab}_{j, \left \langle \mathbf{v}_i \right \rangle_1[j]}\}_{j\in[\kappa]}$ to $\mathcal{A}$.
\end{itemize}
\item[3.] Authentication Phase 1:
\begin{itemize}
\item $\mathtt{Sim_m}$ sets $\delta=(\mathbf{v}_i||sign(\mathbf{v}_i))$.
\item For $i\in[2\kappa]$, $\mathtt{Sim_m}$ sets $\hat{\mathtt{lab}_{i, j}^{out}}=\mathtt{lab}_{i, \delta[i]}^{out}$ if $j=\delta[i]$.
\item For $i\in[2\kappa]$, if $j=1-\delta[i]$, $\mathcal{S}$ sets $\hat{\mathtt{lab}_{i, j}^{out}}$ as  a random value chosen from $\{0, 1\}^{\lambda}$ uniformly, where first bit of  $\hat{\mathtt{lab}_{i, 1-\delta[i]}^{out}}$ is $1-\varsigma_{i, \delta[i]}$.
\item $\mathtt{Sim_m}$ computes and sends $\{ct_{i,j}, \hat{ct_{i, j}}\}_{i\in [\kappa], j\in\{0, 1\}}$ to $\mathcal{A}$ using $\hat{\mathtt{lab}_{i, j}^{out}}_{i\in [2\kappa], j\in\{0, 1\}}$. This process is same as  in $\mathtt{Real}$ execution using  ${\mathtt{lab}_{i, j}^{out}}_{i\in [2\kappa], j\in\{0, 1\}}$.
\end{itemize}
\item[4.]Local Computation Phase: The execution of this phase is indistinguishable from $\mathtt{Real}$ since no information needs to be exchanged between $\mathtt{Sim_m}$ and $\mathcal{A}$.
\item[5.]  Authentication Phase 2:
\begin{itemize}
\item The execution is identical to $\mathtt{Real}$.
\end{itemize}
    \end{itemize}

$\mathtt{Hyb_2}$: We will make four changes to $\mathtt{Hyb_1}$ to obtain $\mathtt{Hyb_2}$, and argue that $\mathtt{Hyb_2}$ is indistinguishable from $\mathtt{Hyb_1}$ from the adversary's view. To be precise, let $\mathtt{GCEval}(\mathtt{GC}, \{\mathtt{\tilde{lab}}_{i}^{in}\}_{i\in[2\kappa]})\rightarrow \{(\tilde{\varsigma}_{i}||\tilde{\vartheta}_{i})_{i\in[2\kappa]}= \{\mathtt{\tilde{lab}}_{i}^{out}\}_{i\in[2\kappa]}\}$. First, we have $\{\mathtt{\tilde{lab}}_{i}^{out}= \mathtt{{lab}}_{i, \delta[i]}^{out}\}_{i\in[2\kappa]}$ based on the correctness of garbled circuits.  Second, we note that   ciphertexts $\{ct_{i,1-\tilde{\varsigma}_{i}}, \hat{ct}_{i, 1-\tilde{\varsigma}_{i+\kappa}}\}_{i\in [\kappa]}$  are computed by exploiting the ``other" set of output labels picked uniformly in $\mathtt{Hyb_1}$. Based on this observation,  $\mathtt{Sim_m}$ actually can directly sample them uniformly at random.  Third, in real execution,  for every $i\in [\kappa]$ and $j\in \{0, 1\}$, $P_1$ sends $ct_{i, \varsigma_{i, j}}$ and $\hat{ct}_{i, \varsigma_{i+\kappa, j}}$ to $P_0$, and then $P_0$ computes $c_i$, $d_i$ and $e_i$ based on them. To simulate this,  $\mathtt{Sim_m}$ only needs to uniformly select random values $c_i$, $d_i$ and $e_i$ which satisfy $\langle \alpha \mathbf{v}_i\rangle_0=(-\sum_{j\in[\kappa]}c_j2^{j-1})$, $\langle sign(\mathbf{v}_i)\rangle_0=(-\sum_{j\in[\kappa]}d_j2^{j-1})$ and $\langle \alpha sign(\mathbf{v}_i)\rangle_0=(-\sum_{j\in[\kappa]}e_j2^{j-1})$. Finally, since $\langle \alpha \mathbf{v}_i\rangle_0$, $\langle sign(\mathbf{v}_i)\rangle_0$ and $\langle \alpha sign(\mathbf{v}_i)\rangle_0$ are part the outputs of functionality $\mathcal{F}_{ONlin}$, $\mathtt{Sim_m}$  can obtain these as the outputs from $\mathcal{F}_{ONlin}$.  In summary,  with the above changes, $\mathtt{Sim_m}$ no longer needs  $\alpha$ of $P_1$. We provide the formal description of $\mathtt{Hyb_2}$ as follows.
 \begin{itemize}
    \item[1.] $\mathtt{Sim_m}$  receives $\left \langle \mathbf{v}_i \right \rangle_0$ from $\mathcal{A}$ as the input of OT$_{\lambda}^{\kappa}$.
\item[2.] Garbled Circuit Phase: Same as $\mathtt{Hyb_1}$.
\item[3.] Authentication Phase 1:
\begin{itemize}
\item $\mathtt{Sim_m}$  runs $\mathtt{GCEval}(\mathtt{GC}, \{\mathtt{\tilde{lab}}_{i}^{in}\}_{i\in[2\kappa]})\rightarrow \{(\tilde{\varsigma}_{i}||\tilde{\vartheta}_{i})_{i\in[2\kappa]}= \{\mathtt{\tilde{lab}}_{i}^{out}\}_{i\in[2\kappa]}\}$.
\item $\mathtt{Sim_m}$  learns $\langle \alpha \mathbf{v}_i\rangle_0$, $\langle sign(\mathbf{v}_i)\rangle_0$ and $\langle \alpha sign(\mathbf{v}_i)\rangle_0$ by sending  $\langle\mathbf{v}_i\rangle_0$ to $\mathcal{F}_{ONlin}$.
\item For $j\in[\kappa]$, $\mathtt{Sim_m}$ uniformly selects random values $c_j$, $d_j$ and $e_j\in \mathbb{F}_{p}$ which satisfy $\langle \alpha \mathbf{v}_i\rangle_0=(-\sum_{j\in[\kappa]}c_j2^{j-1})$, $\langle sign(\mathbf{v}_i)\rangle_0=(-\sum_{j\in[\kappa]}d_j2^{j-1})$ and $\langle \alpha sign(\mathbf{v}_i)\rangle_0=(-\sum_{j\in[\kappa]}$ $e_j2^{j-1})$.
\item  For every $i\in [\kappa]$, $\mathtt{Sim_m}$ computes $ct_{i, \tilde{\varsigma}_{i}}=c_i \oplus \mathbf{Trun}_\kappa(\tilde{\vartheta}_{i})$  and $\hat{ct}_{i, \tilde{\varsigma}_{i+\kappa}}=(d_i||e_i)\oplus\mathbf{Trun}_{2\kappa}(\tilde{\vartheta}_{i+\kappa})$. For ciphertexts $\{ct_{i,1-\tilde{\varsigma}_{i}}, \hat{ct}_{i, 1-\tilde{\varsigma}_{i+\kappa}}\}_{i\in [\kappa]}$, $\mathtt{Sim_m}$    samples them uniformly at random.
\item $\mathtt{Sim_m}$  sends $\{ct_{i,j}, \hat{ct_{i, j}}\}_{i\in [\kappa], j\in\{0, 1\}}$ to $\mathcal{A}$.
\end{itemize}
\item[4.]Local Computation Phase: The execution of this phase is indistinguishable from $\mathtt{Real}$ since no information needs to be exchanged between $\mathtt{Sim_m}$  and $\mathcal{A}$.
\item[5.]  Authentication Phase 2:
\begin{itemize}
\item The execution is identical to $\mathtt{Real}$.
\end{itemize}
    \end{itemize}

$\mathtt{Hyb_3}$: This hybrid we  remove $\mathtt{Sim_m}$'s dependence on $P_1$'s input $\langle \mathbf{v}_i \rangle_1$. The indistinguishability between $\mathtt{Hyb_3}$ and $\mathtt{Hyb_2}$ stems from the security of the garbled circuit. We provide the formal description of $\mathtt{Hyb_3}$ below.
 \begin{itemize}
    \item[1.] $\mathtt{Sim_m}$  receives $\left \langle \mathbf{v}_i \right \rangle_0$ from $\mathcal{A}$ as the input of OT$_{\lambda}^{\kappa}$.
\item[2.] Garbled Circuit Phase:
\begin{itemize}
\item $\mathtt{Sim_m}$  samples $\mathtt{Garble}(1^\lambda, booln^f)\rightarrow (\tilde{\mathtt{GC}}, \{\hat{\mathtt{lab}}_{i}^{in}\}_{i\in\{\kappa+1, \cdots, 2\kappa\}})$  and sends $\{\hat{\mathtt{lab}}_i\}_{i\in\{\kappa+1, \cdots, 2\kappa\}}$ to $\mathcal{A}$ as the output of  OT$_{\lambda}^{\kappa}$. $\mathcal{S}$ also sends $\tilde{\mathtt{GC}}$ and $\{\hat{\mathtt{lab}}_{i}^{in}\}_{i\in[\kappa]}$ to $\mathcal{A}$.
\end{itemize}
\item[3.] Authentication Phase 1:
\item[4.]Local Computation Phase: The execution of this phase is indistinguishable from $\mathtt{Real}$ since no information needs to be exchanged between $\mathtt{Sim_m}$  and $\mathcal{A}$.
\item[5.]  Authentication Phase 2: Same as $\mathtt{Hyb_2}$, where $\mathtt{Sim_m}$ uses $(\langle \mathbf{v}_i\rangle_0, \tilde{\mathtt{GC}},$ and $ \{\hat{\mathtt{lab}}_{i}^{in}\}_{i\in[2\kappa]})$ to process this phase for $\mathcal{A}$.
    \end{itemize}
\end{proof}

\section{Proof of Theorem 5}
\label{Proof of Theorem 5}

\begin{proof}
Assuming that $P_0$ tampered with any of the inputs it holds during the execution, $q$ can be expressed as follows
$$q=\Delta+\sum_{j\in [t]} \mathbf{r}_j(\mathbf{\rho}_j-\alpha  \mathbf{\tau}_j)+  \sum_{i\in [m-1]}\mathbf{r}_i(\mathbf{\tau}_i-\mathbf{\xi}_i)$$
where $\Delta$ refers to the increment caused by $P_0$'s violation of the protocol.   The above formula can be expressed as a 1-degree polynomial function $Q(\alpha)$ with respect to the variable $\alpha$. It is clear that $Q(\alpha)$ is a non-zero polynomial whenever $P_0$ introduces errors. Further, when $Q(\alpha)$ is a non-zero polynomial, it has at most one root. Hence, over the choice of $\alpha$, the probability that $Q(\alpha)=0$ is at most $1/p$ . Therefore, the probability that $P_1$ aborts is at least $1-1/p$ when $P_0$ cheats.
\end{proof}

\end{document}